\newif\ifcomments\commentstrue
\newif\iflong\longtrue

\newif\iftwocolumns\twocolumnsfalse

\documentclass{article}[10pt]
\usepackage[a4paper]{geometry}
\usepackage{fullpage}
\usepackage{hyperref}
\usepackage{macros}
\usepackage{laetitia_macro}
\usepackage[english]{babel}

\usepackage{amsmath}
\usepackage{graphicx}

\usepackage{amsthm}
 \usepackage{amssymb}

\usepackage{thmtools} 
\usepackage{thm-restate}
\usepackage{wrapfig}

\theoremstyle{definition}
\newtheorem{theorem}{Theorem}[section]
\newtheorem{corollary}{Corollary}[theorem]
\newtheorem{lemma}[theorem]{Lemma}

\newtheorem*{theorem*}{Theorem}
\newtheorem*{corollary*}{Corollary}
\newtheorem*{lemma*}{Lemma}
\newtheorem*{fact*}{Fact}
\newtheorem*{conjecture*}{Conjecture}

\theoremstyle{remark}
\newtheorem{remark}{Remark}[section]
\newtheorem*{remark*}{Remark}

\theoremstyle{definition}
\newtheorem{definition}{Definition}[section]
\newtheorem*{definition*}{Definition}

\theoremstyle{definition}
\newtheorem{example}{Example}[section]
\newtheorem*{example*}{Example}

\theoremstyle{definition}
\newtheorem{proposition}{Proposition}[section]
\newtheorem*{proposition*}{Proposition}

\title{Realisability and Complementability of Multiparty Session Types}

\author{Cinzia {Di Giusto} \and Etienne Lozes \and Pascal Urso}

\date{}

\begin{document}

\maketitle

\begin{abstract}
	% !TEX root =  ../mainPPDP.tex
%!TEX spellcheck = en_GB

%TITLE: Realisability and Complementability of Multiparty Session Types 

%Multiparty session types (MPST) are a type-based approach for specifying message-passing distributed systems.
%MPST rely on the notion of global type specifying the global behaviour and local types which are the projections of the global behaviour onto each local participant. An essential property of global types is implementability, i.e., whether the composition of the local behaviours conforms to those specified by the global type. 
%MPST have been studied extensively but almost solely for peer-to-peer communications. 
%In this paper, we generalise the implementability problem to any communication model.
%First, we show that if a global type is implementable in an arbitrary communication model,
%then it is implementable in the synchronous model. Second, we show that if a global type is implementable in the synchronous model, then it is complementable, in the sense that there exists a global type that describes the complementary behaviour of the original global type.
%Third, we give an algorithm to decide whether a complementable global type, given with an explicit complement, is 
%implementable in a communication model. The algorithm is PSPACE in the size of the global type and its complement, when the communication model is fixed. Finally, we propose a complementation construction for global types with sender driven choice with a linear blowup in the size of the global type.

Multiparty session types (MPST) are a type-based approach for specifying message-passing distributed systems.
They rely on the notion of global type, specifying the global behaviour, and local types, which are the projections of the global behaviour onto each local participant. An essential property of global types is realisability, i.e., whether the composition of the local behaviours conforms to those specified by the global type.
We explore how realisability of MPST relates to their complementability, i.e., whether there exists a global type that describes the complementary behaviour of the original global type.
First, we show that if a global type is realisable with p2p communications, then it is realisable with synchronous communications. Second, we show that if a global type is realisable in the synchronous model, then it is complementable, in the sense that there exists a global type that describes the complementary behaviour of the original global type.
Third, we give an algorithm to decide whether a complementable global type, given with an explicit complement, is 
realisable in p2p. As a side contribution, we propose a complementation construction for global types with sender-driven choice, and more generally commutation-deterministic global types.
\end{abstract}

\section{Introduction}
	% !TEX root =  ../mainPPDP.tex
%!TEX spellcheck = en_GB

The design of communication protocols for concurrent and distributed systems is a central concern in programming languages research. 
Modern systems -- from microservices to IoT networks -- rely on structured sequences of message exchanges between multiple parties. 
In the family of behavioural types~\cite{DBLP:conf/concur/VasconcelosH93}, multiparty session types (MPST) were introduced by Honda et al.~\cite{HondaYC08,10.1145/2873052,DBLP:conf/icdcit/YoshidaG20} and have emerged as a powerful methodology to specify and verify such structured interactions. 
An MPST defines a global type describing the communication protocol as a whole, which can be projected into local types for each participant.% (as opposed to binary session types \cite{10.1007/3-540-57208-2_35,10.1007/3-540-58184-7_118,10.1007/BFb0053567} that only describe binary exchanges).

This type-based approach enables static verification of communication safety (no unexpected messages) and liveness properties (no deadlocks) by ensuring each party's code conforms to its local type.  In essence, MPSTs allow protocol designers to treat communication patterns as first-class abstractions that yield a rigorous framework for building correct distributed applications.

In this work, we focus on the realisability problem, i.e.,
whether the local types projected from a global type can
be recomposed into a concrete system that conforms to the
intended global protocol. In classical MPST theory, protocols are often formulated under the sender-driven choice
assumption: only one designated participant (the sender)
decides between alternative message flow branches for all
parties. While this assumption simplifies theoretical analysis, it proves too restrictive for many real-world protocols. We lift this assumption and build a theory for unrestricted choice.
%Global types are usually defined as terms of the grammar:
%$$
%G ::= \gtlabel{p}{q}{\msg};\ G \mid \mathsf{end} \mid G + G \mid X\mid \mathsf{rec}\ X.\ G
%$$
%with  some additional restrictions on choice composition $G+G'$. A recent trend of research consist in relaxing
%these 
%restrictions on the kind of allowed choices. Usually only 
%directed choice is allowed. In some recent works a more relaxed version is considered:
% sender-driven choice~\cite{DBLP:conf/cav/LiSWZ23}. We follow this trend of research and consider deterministic unrestricted choice composition.
 
Local types are  often defined 
as systems of communicating finite state machines~\cite{BrandZafiropulo}, following the works of Villard~\cite{villard-phd} and Denielou and Yoshida~\cite{DBLP:conf/esop/DenielouY12}: to every global type $G$, one can associate a system of communicating finite state machines (CFSMs) $\projectionof{G}$ whose machines are the projections of $G$ onto each participant.
In some recent works~\cite{DBLP:conf/ecoop/Stutz23},  global types may also be seen as a special case of Higher-order Message Sequence Charts~\cite{DBLP:journals/tcs/AlurEY05} (HMSCs), an automaton model that accepts languages of message sequence charts (MSCs). Adopting such a point of view, every global type $G$ defines a language of MSCs $\existentialmsclanguageof{G}$ which provides
a semantics of $G$. In this work, we study two notions on MPSTs coming from these connections with HMSCs and automata theory: deadlock-free realisability, and complementability.

Deadlock-free realisability, introduced by Alur et al.~\cite{DBLP:journals/tcs/AlurEY05} for HMSCs,
is central to the goal of \emph{correct-by-construction} protocol design underpinning the MPST approach.
Deadlock-free realisability ensures that $\projectionof{G}$ does not suffer from communication mismatches, deadlocks, and behaviours that are not part of $\existentialmsclanguageof{\gt}$. In the literature on MPSTs, this property has appeared under various names, including \emph{implementability} and \emph{projectability}, with some slight variations in the formal definition of this notion. 

The second notion, complementability, asks whether a global type $G$ admits a complementary one $\comp{G}$ that gives  all those behaviours that are not described by $G$, i.e., for all synchronous-like MSCs $M$,
$$
M \in \existentialmsclanguageof{\comp{G}} \quad \Leftrightarrow \quad M \not \in \existentialmsclanguageof{G}.
$$
We also study whether some complementation procedure exists that
explicitly constructs $\comp{G}$ from $G$, and with which complexity.
These problems are  popular in automata theory, e.g., it is well know that pushdown automata are not complementable in general, but become complementable adding a visibility condition~\cite{visiblypushdown};
Büchi automata are complementable~\cite{Buechi62}, and there has been a long quest for efficient complementation procedures and optimisations for subclasses of Büchi automata.

\subsubsection*{Contributions}In this work, we approach the realisability question for two fundamental communication models for implementing MPST protocols: the  peer-to-peer ($\ppmodel$) model, and the synchronous one.
In the $\ppmodel$ model, processes communicate via point-to-point  message passing over FIFO channels.
In the synchronous model, communications are based on rendez-vous (handshakes where send and receive happen simultaneously), representing a globally synchronised message-passing semantics. 
These two models are widely considered as canonical extremes in analysing distributed protocols. 

In this work, we also investigate for the first time the question of the complementability of global types.
Our main contribution is to show that realisable global types form a subclass of the complementable ones.
More in detail, we establish the following results:

\begin{itemize}
\item Every global type realisable in the synchronous communication model is complementable (Theorem~\ref{thm:realisable-complementable}), and the complementation is
effective (based on projection and NFA complementation),  with a doubly exponential blowup in the size of the global type.
\item Realisability is decidable for complementable global types provided an explicit complement is also given as input,
both in the synchronous communication model (Theorem~\ref{thm:decidability-of-implementability-in-synch}) and in the peer-to-peer model 
(Theorem~\ref{thm:decidability-of-implementability-in-p2p}). The latter result is non-trivial and based on recent results on RSC systems~\cite{germerie-phd}.
\item We also show that if a global type is realisable in the peer-to-peer model, then it is also realisable in the synchronous one (Theorem~\ref{thm:pp-realizability-implies-synch-realizability}), which is not obvious, as 
deadlock-freedom is not a property that is monotonic in the communication model (a more permissive communication model may both hide some deadlocks and trigger new ones).
\item Finally, we show that not all global types are complementable (Theorem~\ref{thm:not-all-gt-are-complementable}) and we give procedures to find the complement of certain classes:
 global types with at most three participants are complementable (Corollary~\ref{coro:3-participants-complementable}), 
as well as global types with sender-driven choices (Theorem~\ref{thm:sender-driven-choice-complementation}). 
\end{itemize}

\subsubsection*{Outline.} The paper is organised as follows: Section~\ref{sec:preliminaries} introduces the necessary background on executions, Message Sequence Charts, communication models and communicating finite state machines. 
MPST and realisability are introduced in Section~\ref{sec:global_types}. Then Section \ref{sec:complementability} discusses the issue of complementability, discussing an example of a non-complementable global type as well as several complementation procedures. Finally, Section~\ref{sec:decidability} establishes the decidability of realisability for complementable global types.
Section \ref{sec:concl} concludes with some final remarks and discusses related works.
\iflong\else
Omitted proofs can be found in~\cite{DiGiustoLozesUrsoPPDPlong}.
\fi

	\section{Preliminaries}\label{sec:preliminaries}
	% !TEX root =  ../mainPPDP.tex
%!TEX spellcheck = en_GB

%In this section we introduce some notation to talk about communicating automata, executions and message sequence charts.

We assume a finite set of \emph{processes} $\Procs=\{p,q,\ldots\}$ and a finite set of messages (labels) 
$\Msg=\{\msg,\ldots\}$.
We consider two kinds of actions:  \emph{send actions} that are  of the form 
$\send{p}{q}{\msg}$ and are executed by  process $p$  sending message  $\msg$ to  $q$;
 \emph{receive actions} that are of the form $\recv{p}{q}{\msg}$ and are executed by  $q$ receiving  $\msg$ from  $p$.
We write $\Act$ for the finite set $\Procs\times\Procs\times\{!,?\}\times\Msg$ of all actions, and $\Actp$ for the set of actions 
that can be executed by $p$ (i.e., $\send{p}{q}{\msg}$ or $\recv{q}{p}{\msg}$). 
We omit processes when they are clear from the context and simply write $!\msg$ or $?\msg$ for a send or receive action, respectively.

An \emph{event} $\event$ of a sequence of actions $w\in \Act^*$, is an index $i$ in $\{1,\ldots,\length{w}\}$;
$i$ is a send (resp. receive) event of $w$ if $w[i]$ is a send (resp. receive) action.
We write $\sendeventsof{w}$ (resp. $\receiveeventsof{w}$) for the set of send (resp. receive) events 
of $w$ and $\eventsof{w}=\sendeventsof{w}\cup\receiveeventsof{w}$ for the set of all events of $w$.
When all events are labeled with distinct actions, we identify an event with its action.

\subsubsection*{Executions.}
An execution is a well defined sequence of actions $e\in\Act^*$, where 
a receive action is always preceded by a unique corresponding send action.

\begin{definition}[Execution]
	An \emph{execution} over $\Procs$ and $\Msg$ is a sequence of actions $e\in\Act^*$  
	where an injective function from receive events to send events $\source_{e}:\receiveeventsof{w}\to\sendeventsof{w}$ 
	is such that for each receive event $i$ labeled with $\recv{p}{q}{\msg}$,
	 $\source_{e}(i)$ is labeled with $\send{p}{q}{\msg}$, and  $\source_{e}(i)<i$.
\end{definition}

%For two events $i,j\in\eventsof{e}$, we write $i\hbstrictof{e}{} j$ if $i<j$ (order on natural numbers).
%When all send actions of a sequence of actions $w$ are distinct, there is at most one execution 
%$\execution$ such that $w$ is the sequence of actions of $\execution$, and we often identify $w$ with $\execution$
%and let $\source$ be implicit. 

%An execution $\w_1$ is a \emph{prefix} of $w_2$ if $w_1$ is a prefix of $w_2$ and 
%$\source_1$ is the restriction of $\source_2$ to $\receiveeventsof{w_1}$.
For a set of executions $\mathcal{E}$, we write $\prefixclosureof{\mathcal{E}}$ 
for the set of all prefixes of the executions in $\mathcal{E}$. 
We say that an execution $e_2$ is a \emph{completion} of an execution $e_1$
if $e_1$ is a prefix of $e_2$. 
A \emph{concatenation} $e_1\cdot w$ of an execution
$e_1$ and a sequence of actions $w$ is the execution $e_2 = e_1\cdot w_2$ where 
$e_2$ is a completion of $e_1$ (note that $w$ is not an execution, since it may contain receive events
which sources are in $e_1$).  
The \emph{projection} $\projofon{e}{p}$ of an execution $e$ on a process $p$ 
is the subsequence of actions in $\Actp$.  
A send event $s$ is \emph{matched} if there is a receive event $r$ such that $s=\source(r)$.
An execution $e$ is \emph{orphan-free} if $\source$ is surjective over the send events of $e$, i.e.,
all send events are matched.

\subsubsection*{Communication Models.}
In this paper, we consider   two communication models: 
peer-to-peer ($\ppmodel$) and  synchronous ($\synchmodel$). 
However, as commented  in the conclusions,  this work is part of a bigger more general project and a large amount of the results of the paper can be extended to "any" communication model, or 
only require few assumptions. 
In this perspective, we introduce here a general definition of a communication model and discuss specific results regarding how it can be extended to the general case.
\begin{definition}[Communication model]
	\label{def:communication-model}
	A \emph{communication model} $\acommunicationmodel$ is a set 
	$\executionsofmodel{\acommunicationmodel}$ of executions.
\end{definition}
% It is \emph{prefix-closed} if 
% $\executionsofmodel{\acommunicationmodel}$ is prefix-closed, i.e., for all
% $e,e'\in\executionsofmodel{\acommunicationmodel}$, if $e\prefixorder e'$, then $e\in\executionsofmodel{\acommunicationmodel}$.

In the  synchronous communication model $\synchmodel$, message exchanges can be thought as rendezvous synchronisations.
In other words, an execution $e$ belongs to $\executionsofmodel{\synchmodel}$ if
all send events are immediately followed by their corresponding receive events.
\begin{definition}[$\synchmodel$]
	An execution $\execution=(w,\source) \in \executionsofmodel{\synchmodel}$ if
	for all send event $s\in\sendeventsof{e}$, $s+1$ is a receive event of $e$ and $\source(s+1)=s$.
\end{definition}

In the peer-to-peer communication model $\ppmodel$, messages sent by a process $p$ to $q$ transit over a FIFO channel
that is specific to the pair $(p,q)$: if $p$ sends first $m_1$ then $m_2$ to $q$, 
then $m_2$ cannot overtake $m_1$ in the FIFO channel. In particular:
\begin{itemize}
	\item if $m_1$ is not received, then $m_2$ is not received either;
	\item if both are received, then $m_1$ is received before $m_2$.
\end{itemize}
      
\begin{definition}[$\ppmodel$]\label{def:queue-based-communication-model}
	$\executionsofmodel{\ppmodel}$  is
	the set of executions $e$ such that for any two send events $s_1=\send{p}{q}{\msg_1}$ and 
	$s_2=\send{p}{q}{\msg_2}$ in $\sendeventsof{e}$, 
	with $s_1 < s_2$,
	one of the two holds:
	\begin{itemize}
		\item $s_2$ is unmatched, or
		\item there exist $r_1,r_2$ such that $r_1<r_2$, $\source(r_1)=s_1$, and $\source(r_2)=s_2$.
	\end{itemize}
\end{definition}

Note that, $\executionsofmodel{\synchmodel}\subset\executionsofmodel{\ppmodel}$.
Moreover, if $e$ is an execution in $\ppmodel$, then  $\source_e$ is defined as follows:
the source of the $i$-th receive event of $q$ from $p$ is the $i$-th send event of $p$ to $q$.
If $e$ is an execution in $\synchmodel$, then $\source_e$ is defined as follows: 
for all receive event $i$, $\source_{e}(i)=i-1$.

\subsubsection*{Message Sequence Charts.}
While executions correspond to a total order view of the events occurring in a 
system, message sequence charts (MSC) adopt a distributed and partial order view on the events.
For a tuple $\msc=(w_p)_{p\in\Procs}$, each $w_p\in\Actsp$ 
is a sequence of actions representing  the ones executed by process $p$ according to some total, locally observable, order.
We write $\eventsof{\msc}$  for the set $\{(p,i) \mid p \in \Procs \text{ and } 0 \leq i < \length{w_p}\}$.
The label $\actionof{\event}$ of the event $\event=(p,i)$ is the action $w_p[i]$. The event $\event$
is a send (resp. receive) event if it is labeled with a send (resp. receive) action.
We write $\sendeventsof{\msc}$ (resp. $\receiveeventsof{\msc}$) for the set of send (resp. receive) events of $\msc$; we also write $\messageof{\event}$ for the message sent or received at event $\event$, and 
$\processof{\event}$ for the process executing $\event$. Finally, we write 
$\verticalorderstrict{\event_1}{\event_2}$ if there is a process $p$ and $i<j$ such that
$\event_1=(p,i)$ and $\event_2=(p,j)$.
%\pascal{Notations $\sendevents{\msc}$,  $\messageof{e}$, $\receiveevents{\msc}$ and $\processof{e}$ are used once in the paper. Should we keep them?}

\begin{definition}[MSC]\label{def:msc}
	An {MSC} over $\Procs$ and $\Msg$ is a tuple $\msc = \big((w_p)_{p\in\Procs},\source\big)$
    where
    \begin{enumerate}
        \item for each process $p$, $w_p\in\Actsp$ is a finite sequence of actions;
        \item $\source: \receiveeventsof{\msc} \to \sendeventsof{\msc}$ is 
            an injective function from receive events to send events such that
			for all receive event $\event$ labeled with $\recv{p}{q}{\msg}$,
            $\source(\event)$ is labeled with $\send{p}{q}{\msg}$.
    \end{enumerate}
\end{definition}

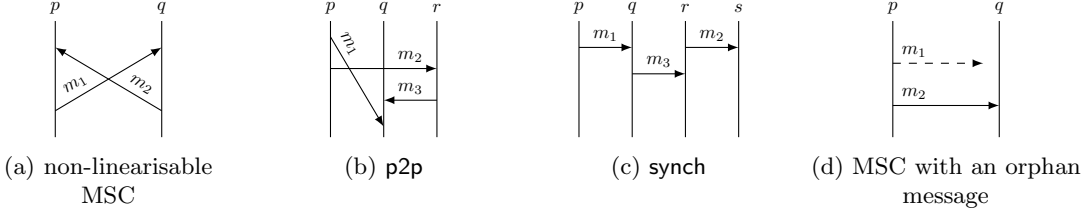
\begin{figure}[t]
		\captionsetup[subfigure]{justification=centering}
	% \centering
	\begin{subfigure}[t]{0.22\textwidth}\centering

		\begin{tikzpicture}[scale=0.7, every node/.style={transform shape}]
			\newproc{0}{p}{-2.2};
			\newproc{2}{q}{-2.2};

			\newmsgm{0}{2}{-1.7}{-0.5}{1}{0.25}{black};
			\newmsgm{2}{0}{-1.7}{-0.5}{2}{0.25}{black};

			\end{tikzpicture}
		\caption{non-linearisable MSC}	\label{fig:raw_msc_ex}

		\end{subfigure}
%        \begin{subfigure}[t]{0.24\textwidth}\centering
%
%            \begin{tikzpicture}[scale=0.7, every node/.style={transform shape}]
%                \newproc{0}{p}{-2.2};
%                \newproc{2}{q}{-2.2};
%    
%                \newmsgm{0}{2}{-0.5}{-1.5}{1}{0.1}{black};
%                \newmsgm{0}{2}{-1}{-1}{2}{0.8}{black};
%    
%                \end{tikzpicture}
%            \caption{$\bagmodel$, and $\erlangmodel$ if $m_1\neq m_2$}	\label{fig:bag_ex}
%    
%            \end{subfigure}
        % \centering
	\begin{subfigure}[t]{0.22\textwidth}\centering
		\begin{tikzpicture}[scale=0.7, every node/.style={transform shape}]
			\newproc{0}{p}{-2.2};
			\newproc{1}{q}{-2.2};
			\newproc{2}{r}{-2.2};

			\newmsgm{0}{1}{-0.3}{-2}{1}{0.15}{black};
			\newmsgm{0}{2}{-0.9}{-0.9}{2}{0.75}{black};
			\newmsgm{2}{1}{-1.5}{-1.5}{3}{0.5}{black};
			% \newmsgm{2}{1}{-2}{-2}{4}{0.5}{black};

			% \newflechevert{Purple}{0}{-0.3}{-0.9};
			% \newflechehor{Purple}{-0.9}{0}{2};
			% \newflechevert{Purple}{2}{-0.9}{-1.5};
		\end{tikzpicture}
		\caption{$\ppmodel$} \label{fig:pp_ex}
	\end{subfigure}
%	\begin{subfigure}[t]{0.24\textwidth}\centering
%		\begin{tikzpicture}[scale=0.7, every node/.style={transform shape}]
%			\newproc{0}{p}{-2.2};
%			\newproc{1}{q}{-2.2};
%			\newproc{2}{r}{-2.2};
%			\newproc{3}{s}{-2.2};
%	
%			\newmsgm{0}{3}{-0.6}{-0.6}{2}{0.8}{black};
%			\newmsgm{2}{3}{-1.1}{-1.1}{3}{0.5}{black};
%			\newmsgm{2}{1}{-1.5}{-1.5}{4}{0.5}{black};
%			\newmsgm{0}{1}{-0.3}{-2}{1}{0.5}{black};
%		\end{tikzpicture}
%		\caption{$\causalmodel$} \label{fig:causal_ex}
%	\end{subfigure}
	% \hfill
%\begin{subfigure}[t]{0.24\textwidth}\centering
%	\begin{tikzpicture}[scale=0.7, every node/.style={transform shape}]
%		\newproc{0}{p}{-2.2};
%		\newproc{1}{q}{-2.2};
%		\newproc{2}{r}{-2.2};
%		\newproc{3}{s}{-2.2};
%
%		\newmsgm{3}{0}{-0.6}{-0.6}{2}{0.2}{black};
%		\newmsgm{3}{2}{-1.1}{-1.1}{3}{0.5}{black};
%		\newmsgm{2}{1}{-1.5}{-1.5}{4}{0.5}{black};
%		\newmsgm{0}{1}{-0.3}{-2}{1}{0.5}{black};
%	\end{tikzpicture}
%			\caption{$\mbmodel$} \label{fig:mb_ex}
%\end{subfigure}
%\begin{subfigure}[t]{0.24\textwidth}\centering
%	\begin{tikzpicture}[scale=0.7, every node/.style={transform shape}]
%		\newproc{0}{p}{-2.2};
%		\newproc{1}{q}{-2.2};
%		\newproc{2}{r}{-2.2};
%
%		\newmsgm{0}{1}{-0.3}{-1.5}{1}{0.15}{black};
%		\newmsgm{1}{0}{-0.9}{-0.9}{2}{0.25}{black};
%		\newmsgm{1}{2}{-2}{-2}{3}{0.5}{black};
%		% \newmsgm{2}{1}{-2}{-2}{4}{0.5}{black};
%	\end{tikzpicture}
%		\caption{$\busmodel$} \label{fig:bus_ex}
%\end{subfigure}
\begin{subfigure}[t]{0.22\textwidth}\centering
	\begin{center}
		\begin{tikzpicture}[scale=0.7, every node/.style={transform shape}]
			\newproc{0}{p}{-2.2};
			\newproc{1}{q}{-2.2};
			\newproc{2}{r}{-2.2};
			\newproc{3}{s}{-2.2};

			\newmsgm{0}{1}{-0.5}{-0.5}{1}{0.5}{black};
			\newmsgm{2}{3}{-0.5}{-0.5}{2}{0.5}{black};
			\newmsgm{1}{2}{-1}{-1}{3}{0.5}{black};

		\end{tikzpicture}
		\caption{$\synchmodel$}
		\label{fig:rsc_ex}
	\end{center}
\end{subfigure}
\begin{subfigure}[t]{0.24\textwidth}\centering
	\begin{tikzpicture}[scale=0.7, every node/.style={transform shape}]
		\newproc{0}{p}{-2.2};
		\newproc{2}{q}{-2.2};

		\newmsgum{0}{2}{-0.8}{1}{0.2}{black};
		\newmsgm{0}{2}{-1.6}{-1.6}{2}{0.2}{black};

		\end{tikzpicture}
	\caption{MSC with an orphan message}	\label{fig:orphan_ex}
\end{subfigure}
		\caption{Examples of MSCs. The sender of a message is at the origin of the arrow and the receiver at the destination. Unmatched send events are depicted with dashed arrows.}\label{fig:exmscs}
\end{figure}

For an execution $\execution$,  $\mscof{\execution}$ is the MSC 
$\big((w_p)_{p\in\Procs},\source\big)$ where $w_p$ is the subsequence of $\execution$ restricted to the actions of $p$,
and $\source$ is the lifting of $\source_{\execution}$ to the events of $(w_p)_{p\in\Procs}$.

\begin{example}
    \label{ex:msc}
     MSC $\msc$ in Fig.~\ref{fig:raw_msc_ex} is an MSC over $\Procs=\{p,q\}$ and $\Msg=\{\msg_1,\msg_2\}$
    with $\msc=\large((w_p,w_q),\source\large)$, $w_p = {}?\msg_2!\msg_1$, $w_q = {}?\msg_1!\msg_2$, $\source((p,0)) = (q,1)$,
    and $\source((q,0)) = (p,1)$. Note that there is no execution $\execution$ such that 
	$\msc=\mscof{\execution}$ as all receptions precede the corresponding sends. On the other hand, the MSC
	in Fig.~\ref{fig:pp_ex} is $\mscof{\execution}$ for the execution $\execution$
	$=\large(!\msg_1!\msg_2?\msg_2!\msg_3?\msg_3?\msg_1,\source\large)$ where 
	$\source(3) = 2$, $\source(5) = 2$ and $\source(6) = 4$. This is the only execution that induces this MSC, but 
	in general there might exist several executions inducing the same MSC.
\end{example}

For a set of processes $\Procs$, 
an MSC $M=\big((w_p)_{p\in\Procs},\source\big)$ is the \emph{prefix} 
of another MSC $M'=\big((w'_p)_{p\in\Procs},\source'\big)$, in short $M \prefixorder M'$,
if for all $p\in\Procs$, $w_p$ is a prefix of $w'_p$ and $\source'(e)=\source(e)$ for all receive events $e$ of $M$.
The \emph{concatenation} of MSCs $M_1$ and $M_2$ is the MSC  $M_1\cdot M_2$ 
obtained by gluing "vertically"  $M_1$ before $M_2$:
formally, if $M_1=((w_p^1)_{p\in\Procs},\source_1)$ and $M_2=((w_p^2)_{p\in\Procs},\source_2)$, 
then $M_1\cdot M_2=((w_p)_{p\in\Procs},\source)$ where
%\begin{inparaenum}[(i)]
 %   \item 
    for all $p$, $w_p=w_p^1\cdot w_p^2$, and
 %   \item 
    $\source$ is defined by $\source(e) = \source_i(e)$ for all receive events $e$ of $M_i$ ($i=1,2$).
%\end{inparaenum}

\subsubsection*{Happens-before relation and linearisations}
%\etienne{TODO cette section et ailleurs: relire en faisant attention à la convention de notation entre ordres partiels $\prec$ et totaux $<$}
In a given MSC $M$,
an event $\event$ happens before  
 $\event'$, if 
(i) $\event$ and $\event'$ are events
of a same process $p$ and happen in that order on 
the time line of $p$,
or (ii) $\event$ is send event matched by $\event'$,
or (iii) a sequence of such situations defines
a path from $\event$ to $\event'$.

\begin{definition}[Happens-before relation]
Let $M$ be an MSC. The happens-before relation over $M$
is the binary relation $\happensbeforestrict$ defined as 
the least transitive relation over $\eventsof{\msc}$ such that:
\begin{itemize}
   \item 
   for all 
   $p,i,j$, if $i<j$, then $(p,i)\happensbeforestrict (p,j)$, and
   \item 
   for all receive events $\event$, $\source(\event) \happensbeforestrict \event$.
\end{itemize}
\end{definition}

\begin{definition}[Linearisation]
	\label{def:linearisation}
	A \emph{linearisation} of an MSC $\msc$ is a
total order $\alinearisation$ on $\eventsof{\msc}$
	that refines $\happensbeforestrict$:  for all events $\event,\event'$, 
	if $\event\happensbeforestrict \event'$, then $\event\alinearisation \event'$. 
\end{definition}
We write $\linearisationsof{\msc}{}$ for the set of all linearisations of $\msc$.
We often identify a linearisation
with the execution it induces. 

\begin{example}
	\label{ex:linearisation}
	Let $\msc$ be the MSC in Fig.~\ref{fig:rsc_ex}. 
	Then $!m_1$ happens before $?m_1$, which
	happens before $!m_3$, and both $!m_3$ and $!m_2$
	happen before $?m_3$.
	Moreover, $\happensbeforestrict$ is a partial order, 
	and 
	$!m_1!m_2?m_1!m_3?m_3?m_2 \in \linearisationsof{\msc_c}{}$.
	On the other hand, consider the MSC $M$ in 
	in Fig.~\ref{fig:raw_msc_ex}; then
	$\happensbeforestrictof{\msc'}$ is not a partial order, because
	$?m_2\happensbeforestrictof{\msc'}!m_1\happensbeforestrictof{\msc'}?m_1\happensbeforestrictof{\msc'}!m_2\happensbeforestrictof{\msc'}?m_2$,
	therefore $\linearisationsof{\msc'}{}=\emptyset$.
\end{example}

Given an MSC $\msc$, we write 
$\linearisationsof{\msc}{\acommunicationmodel}$
to denote 
$\linearisationsof{\msc}{}\cap\executionsofmodel{\acommunicationmodel}$; 
the executions of $\linearisationsof{\msc}{\acommunicationmodel}$
are called the linearisations of $\msc$ 
in the communication model $\acommunicationmodel$. 	

\begin{definition}[$\acommunicationmodel$-linearisable MSC]
	\label{def:linearisable-msc}
	An MSC $\msc$ is \emph{linearisable} in a communication model $\acommunicationmodel$ if 
	$\linearisationsof{\msc}{\acommunicationmodel}\neq\emptyset$.
	We write $\mscsetofmodel{\acommunicationmodel}$ for the set of all MSCs linearisable in $\acommunicationmodel$.
\end{definition}

\begin{example}
	\label{ex:msc-linearisable}
	The MSC $M_b$ in Fig.~\ref{fig:pp_ex} is linearisable in $\ppmodel$,
	with $\linearisationsof{M_b}{\ppmodel}=
%	$ 
%	$$
		\{!m_1!m_2?m_2!m_3?m_3?m_1\}
		$.
%	$$
	However, $M_b$ is not linearisable in $\synchmodel$.
	Finally the MSC $M_c$ in Fig.~\ref{fig:rsc_ex} is linearisable in $\synchmodel$ with
	$\linearisationsof{M_c}{\synchmodel}=$
	$$
	\{~!m_1?m_1!m_2?m_2!m_3?m_3~,~
	   !m_2?m_2!m_1?m_1!m_3?m_3~\}
	$$
\end{example}

Finally, we introduce a property that will be helpful in the next paragraph 
for
giving an alternative characterisation of deadlock-freedom of a system of communicating finite state machines.

\begin{definition}[Causally-closed communication model]\label{def:causally-closed-communication-model}
	A communication model $\acommunicationmodel$ is \emph{causally-closed} if for all MSCs $M$,
	$\linearisationsof{\msc}{\acommunicationmodel}\neq\emptyset$ implies that
	$\linearisationsof{\msc}{\acommunicationmodel}=\linearisationsof{\msc}{}$.
\end{definition}

  % !TEX root =  ../mainPPDP.tex
%!TEX spellcheck = en_GB
Observe that not all communication models are causally closed. It is the case for $\ppmodel$,  but it is immediate  to see that the property is not valid  for $\synchmodel$. Take for instance  MSC $M$ in Fig.~\ref{fig:rsc_ex},  its linearisation 
$!m_1!m_2?m_1?m_2!m_3?m_3$ 
  does not belong to $\linearisationsof{M_c}{\synchmodel}$.  To show that $\ppmodel$ is causally closed, 
notice that $<$ can be replaced with $\porderof{M}$ in  Definition~\ref{def:queue-based-communication-model}.
\begin{lemma}\label{lem:reformulation-of-ppmodel-def}
    Let $e$ be an execution and  $M=\mscof{e}$.
    Then $e$ is an execution of $\ppmodel$ if and only if
    for any two send events $s_1=\send{p}{q}{\msg_1}$ and $s_2=\send{p}{q}{\msg_2}$ in $\sendeventsof{e}$, with $s_1 \porderof{M} s_2$,
        one of the two holds:
        \begin{itemize}
            \item $s_2$ is unmatched, or
            \item there exist $r_1,r_2$ such that $r_1 \porderof{M} r_2$, $\source(r_1)=s_1$, and $\source(r_2)=s_2$.
        \end{itemize}
\end{lemma}
\begin{proof}
    Assume that $e$ is a $\ppmodel$ execution.
    Let $s_1=\send{p}{q}{\msg_1}$ and $s_2=\send{p}{q}{\msg_2}$ be two send events in $\sendeventsof{e}$ such that $s_1 \porderof{M} s_2$.
    Then $s_1<s_2$ in $e$, because $<$ is a linearisation of $\porderof{M}$.
    By  definition of $\ppmodel$, $s_2$ is unmatched or there exists $r_1,r_2$ such that $r_1 < r_2$, $\source(r_1)=s_1$, and $\source(r_2)=s_2$.
    In the second case, $r_1 < r_2$ implies that
    $r_1 \porderof{M} r_2$, because both $r_1$ and $r_2$ occur on the same process $q$.
    Conversely, assume that $e$ and $M$ satisfy the above reformulation of the definition of $\ppmodel$. 
    Let $s_1=\send{p}{q}{\msg_1}$ and $s_2=\send{p}{q}{\msg_2}$ be two send events in $\sendeventsof{e}$ such that $s_1 < s_2$.
    Then $s_1 \porderof{M} s_2$
    because $s_1$ and $s_2$ occur on the same process $p$.
    By the reformulation of the definition of $\ppmodel$, $s_2$ is unmatched or there exists $r_1,r_2$ such that $r_1 \porderof{M} r_2$, $\source(r_1)=s_1$, and $\source(r_2)=s_2$.
    In the second case, $r_1 \porderof{M} r_2$
    implies that $r_1 < r_2$ because $<$ is a  linearisation of $\porderof{M}$.
\end{proof}

\begin{restatable}[]{lemma}{lemppiscasuallyclosed}\label{lem:pp-is-causally-closed}
	$\executionsofmodel{\ppmodel}$ is causally-closed.
\end{restatable}
\begin{proof}
    Let $M\in\mscsetofmodel{\ppmodel}$ 
    and  $e$ be a linearisation of $M$.
    We want to show that $e\in\executionsofmodel{\ppmodel}$.
    By definition of $\mscsetofmodel{\ppmodel}$, 
    there is an execution 
    $e'\in\executionsofmodel{\ppmodel}$
    such that $\mscof{e'}=M$.
    By Lemma~\ref{lem:reformulation-of-ppmodel-def},
    $e$ is also a $\ppmodel$ execution.
\end{proof}
  %The proof of the following Lemma is in Appendix \ref{app:pp-is-causally-closed}.

\subsubsection*{Communicating finite state machines.}
We assume standard notations for automata, words and languages. As usual, a non-deterministic finite state automaton (NFA) is a tuple
$\mathcal A = (Q,\Sigma, \delta, l_0, F)$ where $Q$ is a set of control states, 
$\Sigma$ is an alphabet, $\delta:Q\times\Sigma\to 2^Q$ 
is the transition relation,
$l_0$ is the initial control state, and $F\subseteq Q$ is the set of accepting states.
The language $\languageofnfa{\mathcal A}$ of an NFA $\mathcal A$ and the notion of deterministic
finite state automaton (DFA) or $\varepsilon$ transitions are defined as usual.
We write $\acceptcompletion{\mathcal A}{}$ for the automaton obtained from $\mathcal A$ by setting
$F=Q$. We recall the definition of communicating finite state machines~\cite{BrandZafiropulo}.

\begin{definition}[CFSM] 
    A communicating finite state machine (CFSM) is an NFA with $\varepsilon$-transitions $\acfsm$ over the alphabet $\Act$.
    A system of CFSMs is a tuple $\cfsms = (\acfsm_p)_{p\in\Procs}$.
\end{definition}

 Given a system of CFSMs $\cfsms=(\acfsm_p)_{p\in\Procs}$,
we write $\acceptcompletion{\cfsms}{}$ for the system of CFSMs $\acceptcompletion{\cfsms}=(\acceptcompletion{\acfsm_p})_{p\in\Procs}$
where all states are accepting, i.e., $F_p = Q_p$.

\begin{definition}[Executions of a CFSMs in $\acommunicationmodel$]\label{def:executions-of-cfsms}
Given a system 
$\cfsms = (\acfsm_p)_{p\in\Procs}$ of CFSMs, and a communication model $\acommunicationmodel$,
 $\executionsofcfsms{\cfsms}{\acommunicationmodel}$ is the set  of \emph{executions} $e\in\executionsofmodel{\acommunicationmodel}$
such that for all $p$, $\projofon{e}{p}$ is in $\languageofnfa{\acfsm_p}{}$.
\end{definition}

\begin{remark}
	Let $\acommunicationmodel$ be a communication model, 
	$\cfsms$  a system of CFSMs, and $e,e'\in\executionsofmodel{\acommunicationmodel}$ such that $\mscof{e}=\mscof{e'}$, 
	then $e\in\executionsof{\cfsms}{\acommunicationmodel}$ 
	if and only if $e'\in\executionsof{\cfsms}{\acommunicationmodel}$. This follows from the fact that $\projofon{e}{p}= \projofon{e'}{p}$ for all $p$.
\end{remark}

We write $\mscsofcfsms{\cfsms}{\acommunicationmodel}$ for the set $\{\mscof{e} \mid e\in\executionsofcfsms{\cfsms}{\acommunicationmodel}\}$.

A system is orphan-free if, whenever all machines have reached an accepting state, no message
remains in transit, i.e., no message is sent but not received.

\begin{definition}[Orphan-free]\label{def:orphan-free}
	A system of CFSMs $\cfsms$ is \emph{orphan-free} for a communication model 
	$\acommunicationmodel$ if
	for all $e\in\executionsofcfsms{\cfsms}{\acommunicationmodel}$,
	$e$ is orphan-free.
\end{definition}

All  synchronous executions are orphan-free by definition.

A system is deadlock-free if, 
any \emph{partial} execution can be extended/completed to an accepting execution.

\begin{definition}[Deadlock-free]\label{def:deadlock-free}
	A system of CFSMs $\cfsms$ is \emph{deadlock-free} in a communication model
	$\acommunicationmodel$ if for all 
	$e\in\executionsofcfsms{\acceptcompletion{\cfsms}}{\acommunicationmodel}$,
	there is an execution $e'\in\executionsofcfsms{\cfsms}{\acommunicationmodel}$ 
	such that $e\prefixorder e'$.
\end{definition}

\begin{remark}\label{rem:equivalent-formulation-of-deadlock-free-cfsms}
	A system of CFSMs $\cfsms$ is \emph{deadlock-free} for a communication model
	$\acommunicationmodel$ if and only if 
	$$
	\executionsofcfsms{\acceptcompletion{\cfsms}}{\acommunicationmodel}\subseteq
	\prefixclosureof{\executionsofcfsms{\cfsms}{\acommunicationmodel}}
	$$
\end{remark}

The following result shows that, for either $\ppmodel$ or $\synchmodel$ communication
models, the deadlock-freedom property of a system of CFSMs can be expressed as a property on the MSCs
of the system. %The proof can be found in Appendix \ref{app:deadlock-free-as-a-property-on-mscs-for-p2p-and-synch}.

\begin{restatable}[Deadlock-freedom as an MSC property]{proposition}{propdeadlockfreeasapropertyonmscsforppandsynch}
	\label{prop:deadlock-free-as-a-property-on-mscs-for-p2p-and-synch}
	Assume $\acommunicationmodel$ is $\ppmodel$ (respectively, $\acommunicationmodel$ is $\synchmodel$).
	Then a system of CFSMs $\cfsms$ is deadlock-free for $\acommunicationmodel$ if and only if
	$$
	\mscsofcfsms{\acceptcompletion\cfsms}{\acommunicationmodel}\subseteq\prefixclosureof{\mscsofcfsms{\cfsms}{\acommunicationmodel}}.
	$$
\end{restatable}
%\input{proofs/prop-deadlock-free-as-a-property-on-mscs-for-p2p-and-synch.tex}

%\iflong
%\include{proofs/prop-deadlock-free-as-a-property-on-mscs-for-p2p-and-synch}
%\else
%The proof of Proposition~\ref{prop:deadlock-free-as-a-property-on-mscs-for-p2p-and-synch} can be found in Appendix~\ref{app:deadlock-free-as-a-property-on-mscs-for-p2p-and-synch}.
%\fi

	\section{Global Types}\label{sec:global_types}
	% !TEX root =  ../mainPPDP.tex
%!TEX spellcheck = en_GB

In this section, we define global types. We deviate from the standard definition to allow for more liberal behaviours. In our setting, global types are automata that describe a language of MSCs. 
%and explain how they can further describe a language of executions when the communication model
%is fixed. We conclude with the definition of implementability of a global type.
An \emph{arrow} is a triple $(p,q,m)\in\Procs\times\Procs\times\Messages$ with $p\neq q$;
we often write $\marrow{p}{q}{m}$ instead of $(p,q,m)$, and write $\labelalphabet$
to denote the finite set of arrows.
Two arrows $\marrow{p_1}{q_1}{m_1}$ and $\marrow{p_2}{q_2}{m_2}$ 
\emph{commute} if $\{p_1,q_1\}\cap\{p_2,q_2\}=\emptyset$.

\begin{definition}[Global Type]
    A global type $\gt$ is a non-deterministic finite state automaton over the alphabet $\labelalphabet$.
\end{definition}

\begin{remark}
    Following a more standard approach,
    we could equally define global types with the following grammar
    $$
    \gt ::= \mathsf{end}~|~\marrow{p}{q}{m};\gt~|~\gt+\gt~|~\mathsf{rec}~X.\gt~|~X
    $$
    and associate a non-deterministic finite state automaton to each global type.
\end{remark}

We also consider \emph{deterministic} global types that are defined as deterministic finite state automata over the alphabet $\labelalphabet$;
we always explicitly stress determinism and otherwise implicitly assume non-determinism.

%% DIRE EN QUOI C'EST PLUS GENERAL
% Notice that our definition is slightly more general then the standard one (such as that in \cite{HondaYC08}). 

%\begin{remark} We do not require global types 
%to have sender-driven choice, as in \cite{DBLP:conf/cav/LiSWZ23}.
%\end{remark}

The projection of a global type $\gt$ on a process $p$ is the CFSM $\gt_p$ obtained by 
replacing each arrow $\marrow{q}{r}{\msg}$ of a transitions
of $\gt$ by the corresponding action of $p$ (either $\send{p}{r}{\msg}$ if $p=q$, or $\recv{q}{p}{\msg}$ if $p=r$, 
or $\varepsilon$ otherwise).

\begin{definition}[Projected system of CFSMs]\label{def:projected-system}
    The projected system of CFSMs $\projectionof{\gt}$ 
    associated to a global type $\gt$ is the tuple $(\gt_p)_{p\in\Procs}$.
\end{definition}

%\etienne{Notion de produit synchrone: faire une definition plus marquante? dire que c'est le global type "naturel" d'un systeme?}
Conversely, for every system, we can associate a global type computing its synchronous product. 
Let $\cfsms=(\mathcal A_p)_{p\in\Procs}$ be a system of CFSMs, where 
$\mathcal A_p=(L_p,\Actp,\delta_p,l_{0,p},F_p)$ is the CFSM associated to process $p$.
The \emph{synchronous product} of $\cfsms$ is the global type $\productof{\cfsms}=(L,\labelalphabet,\delta,l_{0},F)$, 
where $L=\Pi_{p\in\Procs}L_p$, 
$l_0=(l_{0,p})_{p\in\Procs}$, $F=\Pi_{p\in\Procs}F_p$, and
$\delta$ is the transition relation defined by
$(\vec l,\marrow{p}{q}{m},\vec l')\in\delta$ if
$(l_p,\send{p}{q}{m},l'_p)\in\delta_p$,
$(l_q,\recv{p}{q}{m},l'_q)\in\delta_q$,
and $l'_r=l_r$ for all $r\not\in\{p,q\}$,
and $(\vec{l},\varepsilon,\vec l')\in\delta$ if
$(l_p,\varepsilon,l'_p)\in\delta_p$ for some $p\in\Procs$
and $l'_r=l_r$ for all $r\neq p$.

Projection and synchronous product form a Galois connection.
For two NFAs $\A,\B$ over a same alphabet,
let $\A\nfaincl\B$ denote that 
$\languageofnfa{\A}\subseteq\languageofnfa{\B}$;
similarly, for two tuples of NFAs 
$\A=(\A_p)_{p\in\Procs}$ and $\B=(\B_p)_{p\in\Procs}$,
let $\A\nfaincl\B$ denote that
$\languageofnfa{\A_p}\subseteq\languageofnfa{\B_p}$ for all $p\in\Procs$.

\begin{proposition}[Galois connection]\label{prop:galois-connection}
    Let $\cfsms$ be a system of CFSMs and $\gt$ a global type.
    Then
    $$
    \projectionof{\gt}\nfaincl\cfsms 
    \qquad\mbox{if and only if}\qquad
    \gt\nfaincl\productof{\cfsms}.
    $$
\end{proposition}

\iflong
\begin{proof}
    We fix $\cfsms=(\A_p)_{p\in\Procs}$.
    We show the two implications separately, starting from the right to left one.
    \begin{itemize}
        \item Assume that $\gt\nfaincl\productof{\cfsms}$.
            Let $p\in\Procs$, and a sequence of actions $w_p\in\languageofnfa{\gt_p}$;
            by definition of $\gt_p$, there is a sequence of arrows $w\in\languageofnfa{\gt}$ such that 
            $w_p=\projofon{w}{p}$. By hypothesis, $w\in\languageofnfa{\productof{\cfsms}}$. By definition of $\productof{\cfsms}$,
            $w_p\in\languageofnfa{\A_p}$, which shows that $\gt_p\nfaincl\A_p$, and more generally $\projectionof{\gt}\nfaincl{\cfsms}$.
        \item Assume that $\projectionof{\gt}\nfaincl{\cfsms}$.
            Let $w\in\languageofnfa{\gt}$. By definition of $\gt_p$, $w_p\eqdef\projofon{w}{p}\in\languageofnfa{\gt_p}$.
            From the hypothesis, $w_p\in\languageofnfa{\A_p}$ for all $p$.
            Assume $w=a_1\ldots a_n$ and let $a_{i,p}=\projofon{a_i}{p}\in\Act\cup\{\varepsilon\}$.
            Since $w_p$ is accepted by $\A_p$, there is a "run"
            $$
            \rho_p \quad \eqdef \quad q_{0,p}~a_{1,p}~q_{1,p}~a_{2,p}~\ldots~q_{n-1,p}~a_{n,p}~q_{n,p}
            $$
            such that 
            \begin{itemize}
                \item $q_{0,p}$ is the initial state of $\A_p$, $q_{n,p}$ is an accepting state of $\A_p$,
                \item if $a_{i}=\varepsilon$, then $q_{i,p}=q_{i-1,p}$, 
                \item if $a_{i}\in\Act$, there is a path from $q_{i-1}$ to $q_{i}$ labeled with $\varepsilon^*a_{i,p}\varepsilon^*$ that may include some $\varepsilon$ steps of $p$
            \end{itemize}
            There is therefore a run of $w$ on $\productof{\cfsms}$
            $$
            (q_{0,p})_{p\in\Procs}\quad a_{1}\quad(q_{1,p})_{p\in\Procs}\quad a_{2}\quad \ldots\quad (q_{n-1,p})_{p\in\Procs}\quad a_{n}\quad (q_{n,p})_{p\in\Procs}
            $$            
            defined as a synchronous product of all $\rho_p$ (while reading $a_i$, we schedule first the silent steps, then the synchronised send and receive actions of the two processes involved in $a_i$, then the remaining silent steps).
            It is accepting by construction, since each $\rho_p$ is accepting, so $w\in \languageofnfa{\productof{\cfsms}}$, which shows that $\gt\nfaincl\productof{\cfsms}$.
    \end{itemize}
\end{proof}
\fi

% The synchronous product $\preproductof{\cfsms}$ is a non-deterministic global type. 
% Its determinisation $\productof{\cfsms}\eqdef\detof{\preproductof{\cfsms}}$ by standard powerset construction is a  global type.

We write $\labeltoexec{w}$ to denote %the synchronous execution coded by the sequence of arrows $w$, i.e., 
the execution obtained by replacing each arrow $\marrow{p}{q}{m}$ of $w$ with 
 $\send{p}{q}{m}\cdot\recv{p}{q}{m}$.
We write $\labeltomsc{w}$ to denote the MSC $\mscof{\labeltoexec{w}}$.

A global type defines a language of MSCs in two different ways, one existential and one universal.
Let $\labellanguageof{\gt}$ be the set of 
sequences of arrows $w$ accepted by $\gt$.
Note that for 
$w\in\Arrows^*$, the function $w \mapsto\labeltomsc{w}$ with $ \labeltomsc{w}\in\mscsetofmodel{\synchmodel}$ is not injective, as two arrows with disjoint pairs of processes
commute. We write $w_1\sim w_2$ if 
$\labeltomsc{w_1}=\labeltomsc{w_2}$, and
$[w]$ for the equivalence class of $w$ with respect to $\sim$.
The existential MSC language $\existentialmsclanguageof{\gt}$ of a global type $\gt$ is the set 
of MSCs that admit at least one representation as a sequence of arrows in $\labellanguageof{\gt}$,
and the universal MSC language $\universalmsclanguageof{\gt}$ of a global type $\gt$ is the set
of MSCs whose representations as a sequence of arrows are 
all in $\labellanguageof{\gt}$:
$$
\begin{array}{rl}
\existentialmsclanguageof{\gt}\eqdef & \{\labeltomsc{w}\mid w\in\labellanguageof{\gt}\}\\
\universalmsclanguageof{\gt}\eqdef & \{\labeltomsc{w}\mid [w]\subseteq\labellanguageof{\gt}\}.
\end{array}
$$

\begin{remark}\label{rem:existential-universal}
    Note that the two languages may differ even for deterministic global types,
    and do not confuse the existential and universal acceptance condition
    with the angelic and demoniac \emph{branching} non-determinism of
    NFAs. Here $\exists$ and $\forall$ refer
    to angelic and demoniac \emph{commutative} non-determinism,
    but the branching non-determinism is angelic in both cases.
    In particular, we have
    $$
    \existentialmsclanguageof{\detof{\gt}} 
     = \existentialmsclanguageof{\gt}
    \quad\text{and}\quad
    \universalmsclanguageof{\detof{\gt}}
    = \universalmsclanguageof{\gt}\iflong.\fi
    $$
    \iflong\else
    where $\detof{\gt}$ is the deterministic global type obtained
    \fi
\end{remark}

\begin{definition}[Commutation-closed]
    A global type $\gt$ is \emph{commutation-closed} if 
    $
    \existentialmsclanguageof{\gt}=\universalmsclanguageof{\gt}.
    $
\end{definition} 
In that case, we write $\msclanguageofcc{\gt}{}$ for the common language. 
As an example, any global type with $\cardinalof{\Procs}\leq 3$ is commutation-closed, as 
any two arrows share at least one process and therefore do not commute.

%\begin{definition}\label{def:commutation-closed-global-type}
%  A global type $\gt$ is \emph{commutation-closed} if 
%$\existentialmsclanguageof{\gt}=\universalmsclanguageof{\gt}$. 
%\end{definition}

%\cinzia{remove proposition and make it as plain text and move it to section 6 when it is needed for the decidability}
% \begin{proposition}\label{prop:three-machines-implies-commutation-closed}
%     Let $\gt$ be global type with $\#\Procs\leq 3$ and no machine sending messages 
%     to itself.
%     Then $\gt$ is  commutation-closed.
% \end{proposition}
% \begin{proof}
%     It follows easily by observing that there are no commutative arrows.
% \end{proof}

% \cinzia{never used, to remove}
% A MSC language $L\subseteq\mscsetofmodel{\synchmodel}$ is \emph{recognisable} if there is
% a commutation-closed global type $\gt$ such that $L=\msclanguageofcc{\gt}$.

\begin{proposition}\label{prop:universal-existential-synch-inclusion}
    For all global type $\gt$, 
    $$
    \universalmsclanguageof{\gt} \subseteq \existentialmsclanguageof{\gt} \subseteq \msclanguageofsystem{\projectionof{\gt}}{\synchmodel}.
    $$
\end{proposition}

\begin{proof}(Sketch)
    The first inclusion is immediate from the definitions.
    For the second inclusion, let $M\in \existentialmsclanguageof{\gt}$, we show that $M\in\msclanguageofsystem{\projectionof{\gt}}{\synchmodel}$.
    By definition of $\existentialmsclanguageof{\gt}$, there is a word $w\in\languageofnfa{\gt}$ such that
    $\mscof{w}=M$. Let $\rho$ be an accepting run of $\gt$ for $w$. For every $p\in\Procs$, let $\rho_p$ denote the run of $\gt_p$ (the machine of $p$ in 
    $\projectionof{\gt}$) obtained by projecting $\rho$: note that we kept $\varepsilon$-transitions in $\gt_p$, see Definition~\ref{def:projected-system},
    so $\rho_p$ is obviously defined. Then $\rho_p$ is an accepting run of $\gt_p$, therefore 
    $M\in\msclanguageof{\gt}{\synchmodel}$ (by Definition~\ref{def:executions-of-cfsms}).
\end{proof}

\begin{restatable}{lemma}{lemproductofgt}\label{lem:product-of-gt}
    Let $\gt_1$ and $\gt_2$ be two global types, with $\gt_2$ commutation-closed. Then
    $$
    \existentialmsclanguageof{\gt_1\otimes\gt_2}=\existentialmsclanguageof{\gt_1}\cap\msclanguageofcc{\gt_2}.
    $$
\end{restatable}

\begin{proposition}\label{prop:product-is-commutation-closed}
    For all system $\cfsms$ of communicating finite state machines, 
    $\productof{\cfsms}$ is commutation-closed and
    $$
    \msclanguageof{\cfsms}{\synchmodel} = \msclanguageofcc{\productof{\cfsms}}.
    $$
\end{proposition}

\begin{proof}(Sketch)
    The runs of the synchronous executions $$!m_1?m_1\ldots!m_k?m_k$$ of $\cfsms$ are in one-to-one correspondence with the 
    runs of the words $m_1\ldots m_k$ of $\productof{\cfsms}$.
\end{proof}

When a global type is implemented in a concrete system, its behaviour  depends on the chosen communication model. 

\begin{definition}[Global Type Language]\label{def:global-type-language}
    Let $\gt$ be a global type and $\acommunicationmodel$ a communication model. 
    The language of $\gt$ in $\acommunicationmodel$ is $
    \executionsof{\gt}{\acommunicationmodel}\eqdef
    \bigcup\{\linearizationsof{M}{\acommunicationmodel}\mid M\in\existentialmsclanguageof{\gt}\}
    $.
\end{definition}

\subsubsection*{Realisability}\label{sec:realisability}

We have finally  collected all the definitions needed to introduce our notion of realisability
of global types that is parametric in a given communication model.
\begin{definition}[Deadlock-free realisability]\label{def:realisability}
    A global type $\gt$ is \emph{deadlock-free realisable}\footnote{We sometimes say simply \emph{realisable} instead of \emph{deadlock-free realisable}.}
    in the communication model
    $\acommunicationmodel$
    if there is a system $\cfsms$ such that the following two conditions hold:
    \begin{enumerate}
    \item $\executionsof{\cfsms}{\acommunicationmodel} = 
                \executionsof{\gt}{\acommunicationmodel}$
    \item $\cfsms$ is deadlock-free in $\acommunicationmodel$.
    \end{enumerate}
\end{definition}

Condition~1 of Definition~\ref{def:realisability} corresponds to  \emph{global type conformance}: the executions of the system do not 
deviate from the ones specified by the global type.

Deadlock-free realisability is equivalent to the notion of 
\emph{safe realisability} of~\cite{DBLP:journals/tcs/AlurEY05} 
when $\acommunicationmodel$ is $\ppmodel$ or $\synchmodel$.
This is not the case for other communication models, like mailbox, for which Proposition~\ref{prop:deadlock-free-as-a-property-on-mscs-for-p2p-and-synch}
does not hold. However, even for such communication models,
global type conformance can still be characterized in terms of MSCs.

\begin{proposition}[Global type conformance as an MSC property]%{proposition}{propmscversionofcond1ofrealizability}
    \label{prop:msc-version-of-cond1-of-realizability}
    Let $\executionsofmodel{\acommunicationmodel}\supseteq\executionsofmodel{\synchmodel}$, 
   Condition~1 of Definition~\ref{def:realisability} is equivalent to
    $\mscsofcfsms{\cfsms}{\acommunicationmodel} \subseteq \existentialmsclanguageof{\gt}$.
\end{proposition}
    
We conclude with a classical result, briefly stated 
in~\cite{DBLP:journals/tcs/AlurEY05}, that establishes that
$\projectionof{\gt}$ is the best system candidate
for realising $\gt$. We say that an NFA is \emph{pruned}
if every state belongs to at least one accepting path. Note that we can always prune an NFA without changing its language.

\begin{proposition}[Canonicity of projection]\label{prop:canonicity-of-projection}
    Assume that $\executionsofmodel{\acommunicationmodel}\supseteq\executionsofmodel{\synchmodel}$
    and $\gt$ is a pruned global type. Then 
    $\gt$ is deadlock-free realisable in $\acommunicationmodel$ iff the two conditions of Definition~\ref{def:realisability} hold for $\cfsms=\projectionof{\gt}$.
\end{proposition}

\begin{proof}
    The only-if direction is immediate.
    Assume $\cfsms$ satisfies the two conditions of Definition~\ref{def:realisability}. As
    $\cfsms$ is deadlock-free, we can assume without loss of generality that
    all machines of $\cfsms$ are pruned. 
    By condition~1, we have
    $\existentialmsclanguageof{\gt}\subseteq\mscsofcfsms{\cfsms}{\acommunicationmodel}$; since moreover 
    $\existentialmsclanguageof{\gt}\subseteq \mscsetofmodel{\synchmodel}$, we get
    $\existentialmsclanguageof{\gt}\subseteq\mscsofcfsms{\cfsms}{\synchmodel}=\existentialmsclanguageof{\productof{\cfsms}}$;
    since $\productof{\cfsms}$ is commutation-closed (Proposition~\ref{prop:product-is-commutation-closed}), 
    we get $\gt\nfaincl\productof{\cfsms}$, and by Galois connection (Proposition~\ref{prop:galois-connection}),
    $$\projectionof{\gt}\nfaincl\cfsms.$$
    From there, we get $\executionsof{\projectionof{\gt}}{\acommunicationmodel} \subseteq \executionsof{\cfsms}{\acommunicationmodel}=\executionsof{\gt}{\acommunicationmodel}$, which shows condition~1.
    For condition~2, since both $\gt$ and $\cfsms$ are pruned, 
    $\projectionof{\acceptcompletion{\gt}}\nfaincl\acceptcompletion{\cfsms}$,
    and we get
    $$
    \executionsof{\projectionof{\acceptcompletion{\gt}}}{\acommunicationmodel} 
    \subseteq 
    \executionsof{\acceptcompletion{\cfsms}}{\acommunicationmodel} 
    \subseteq \prefixclosureof{\executionsof{\gt}{\acommunicationmodel}}.
    $$
\end{proof}

In the remainder, we always implicitly assume that global types are pruned, and that deadlock-free realisability is defined 
according to Definition~\ref{def:realisability} with $\cfsms=\projectionof{\gt}$.

%\etienne{Explain the difference and why it matters}

	\section{Complementability}\label{sec:complementability}
	% !TEX root =  ../mainPPDP.tex
%!TEX spellcheck = en_GB

In this section, we introduce the notion of complementability for global types
and we present some complementation procedures.
Our complementation procedures 
target for some notable subclasses of global types,
namely
commutation-closed global types, $\synchmodel$-realisable global-types, and global types
with sender-driven choice.
For the later, we introduce a generalisation of
sender-driven choice that we call commutation-determinism.
Our complementation procedure naturally fits in this larger framework.

We start with the definition of complement global type:

\begin{definition}[Complement of a  Global Type]
    A global type $\comp{\gt}$ is a complement of $\gt$ 
    if $\existentialmsclanguageof{\gt}=
     \mscsetofmodel{\synchmodel}
    \setminus
    \existentialmsclanguageof{\comp{\gt}}$.
    We say that $\gt$ is {\em complementable} if it admits at least one complement.
\end{definition}

Note that $\gt$ is complementable if and only if $\detof{\gt}$ is
complementable, since the determinisation does not change the MSC language
(see Remark~\ref{rem:existential-universal}).

\begin{theorem}\label{thm:not-all-gt-are-complementable}
    Not all global types are complementable.
\end{theorem}

\begin{proof}
    Let $\Procs=\{p,q,r,s\}$, $\Messages=\{m_1,m_2,m_3\}$, and 
    $\Arrows=\{\gtlabel{p}{q}{m_1},\gtlabel{r}{s}{m_2},\gtlabel{p}{q}{m_3}\}$.
    Consider a\footnote{for instance take the minimal DFA recognising this language} global type $\gt_0$ such that $\languageofnfa{\gt_0}=(m_1+m_2)^*(m_2+m_3)^*$. 
    Note that $\gt_0$ is commutation-closed and the MSCs of $\existentialmsclanguageof{\gt_0}$ are of the form depicted in Fig.~\ref{fig:msc-regular}.

    Now consider the global type $\gt$ depicted in 
    Fig.~\ref{fig:non-complementable-gt-1}. 
    We claim that for all natural numbers $k_1,k_2,k_3$,
    \iftwocolumns
    the following two are equivalent:
    \begin{itemize}
        \item $\mscof{m_1^{k_1}m_2^{k_2}m_3^{k_3}} \in \existentialmsclanguageof{\gt}$;
        \item $k_1 > k_2$ or $k_2 > k_3$.
    \end{itemize}
    \else
    $$
    \mscof{m_1^{k_1}m_2^{k_2}m_3^{k_3}} \in \existentialmsclanguageof{\gt}
    \quad \mbox{if and only if} \quad
    k_1 > k_2 ~~\mbox{or}~~ k_2 > k_3.
    $$
    \fi
    The claim follows by analysis on the paths leading to accepting states in $\gt$, as hinted by
    Fig.~\ref{fig:non-complementable-gt-1}.
    For instance, the language of the words accepted by $\gt$ with state $q_8$ as the last state is of the form 
    $m_1^+(m_2m_1)^*m_3^+$, therefore lead to MSCs with $k_1>k_2$.

    We show that $\gt$ is not complementable. By contradiction, suppose that $\gt$ is complementable and $\gt'$ is a complement.
    Let $\gt''=\gt'\otimes \gt_0$; by Lemma~\ref{lem:product-of-gt}, and the fact that $\gt_0$ is commutation-closed,
    \iftwocolumns
    $$
        \begin{array}{rcl}
            \existentialmsclanguageof{\gt''} 
            & = & 
            \msclanguageofcc{\gt_0}{}\cap\existentialmsclanguageof{\gt'}\\
            & = &
            \{\mscof{m_1^{k_1}m_2^{k_2}m_3^{k_3}}~\mid k_1\leq k_2\leq k_3\}.
        \end{array}
    $$
    \else
    $$
        \existentialmsclanguageof{\gt''}~=~\msclanguageofcc{\gt_0}{}\cap\existentialmsclanguageof{\gt'}~=~
        \{\mscof{m_1^{k_1}m_2^{k_2}m_3^{k_3}}~\mid k_1\leq k_2\leq k_3\}.
    $$
    \fi
    Now, let $\mathcal A$ denote the NFA obtained from $\gt''$ after replacing each $m_2$ transition with an $\varepsilon$ transition.
    Then $\languageofnfa{\mathcal A}=\{m_1^{k_1}m_3^{k_3}\mid k_1\leq k_3\}$, which is not a regular language, hence the contradiction.
\end{proof}

\begin{figure*}[ht]
    \begin{tikzpicture}

    \node[state, initial, initial text={}] (0) at (0,2) {$q_0$};
    \node[state, accepting, color=blue] (1) at (2,4) {$q_1$};
    % \node[state, accepting, color=teal] (2) at (2,-2) {$q_2$};
    % \node[state, accepting, color=brown] (3) at (2,0) {$q_3$};
    \node[state, accepting, color=red] (4) at (4,2) {$q_4$};
    \node[state, accepting, color=red] (5) at (6,2) {$q_5$};
    \node[state, accepting, color=blue] (6) at (6,4) {$q_6$};
    \node[state] (7) at (8,4) {$q_7$};
    \node[state, accepting, color=blue] (8) at (4,4) {$q_8$};
    \node[state] (9) at (4,0) {$q_9$};
    \node[state, accepting, color=red] (10) at (2,0) {$q_{10}$};
    % \node[state] (11) at (8,0) {$q_{11}$};
    % \node[state, accepting, color=brown] (12) at (10,0) {$q_{12}$};
    % \node[state] (13) at (6,-2) {$q_{13}$};
    % \node[state, accepting, color=teal] (14) at (8,-2) {$q_{14}$};
    % \node[state, accepting, color=teal] (15) at (10,-2) {$q_{15}$};

    \draw[->] (0) edge node [sloped,above] {$m_1$} (1);
    \draw[->] (0) edge node [sloped,below]{$m_2$} (10);
    % \draw[->] (0) edge node [above] {$m_3$} (3);
    \draw[->, loop below] (1) edge node {$m_1$} (1);
    % \draw[->, loop below] (2) edge node {$m_2$} (2);
    % \draw[->, loop below] (3) edge node {$m_3$} (3);
    % \draw[->] (1) edge node [left] {$m_3$} (3);
    \draw[->] (1) edge node [sloped, above] {$m_2$} (4);
    \draw[->] (1) edge node [above] {$m_3$} (8);
    \draw[->] (4) edge node [sloped, pos=.65, above] {$m_2$} (5);
    \draw[->, loop right] (5) edge node [right] {$m_2$} (5);
    \draw[->] (4) edge node [right] {$m_3$} (9);
    \draw[->] (4) edge node [sloped, above] {$m_1$} (6);
    \draw[->, bend left] (6) edge node [sloped, below] {$m_2$} (7);
    \draw[->, bend left] (7) edge node [sloped, above] {$m_1$} (6);
    \draw[->] (6) edge node [above] {$m_3$} (8);
    \draw[->, loop below] (8) edge node [below] {$m_3$} (8);
    \draw[->, bend left] (9) edge node [above] {$m_2$} (10);
    \draw[->, bend left] (10) edge node [above] {$m_3$} (9);
    \draw[->, loop left] (10) edge node [left] {$m_2$} (10);
    % \draw[->] (3) edge node [above] {$m_2$} (11);
    % \draw[->, bend left] (11) edge node [above] {$m_3$} (12);
    % \draw[->, bend left] (12) edge node [below] {$m_2$} (11);
    % \draw[->] (2) edge node [above] {$m_1$} (13);
    % \draw[->, bend left] (13) edge node [above] {$m_2$} (14);
    % \draw[->, bend left] (14) edge node [below] {$m_1$} (13);
    % \draw[->] (14) edge node [above] {$m_3$} (15);
    % \draw[->, loop right] (15) edge node [right] {$m_3$} (15);

    % on ajoute des accolades
    \draw [decorate,decoration={brace,amplitude=5pt}, color=blue] (9,4.5) -- node [midway, right=1em] {$k_1>k_2$} (9,3.5) ;
    \draw [decorate,decoration={brace,amplitude=5pt}, color=red] (9,2.5) -- node [midway, right=1em] {$k_2>k_3$} (9,-0.7) ;
    % \draw [decorate,decoration={brace,amplitude=5pt}, color=brown] (12,0.6) -- node [midway, right=1em] {$k_3>k_2$} (12,-.6) ;
    % \draw [decorate,decoration={brace,amplitude=5pt}, color=teal] (12,-1.5) -- node [midway, right=1em] {$k_2>k_1$} (12,-2.5) ;

\end{tikzpicture}
    \caption{A non-complementable global type
    \label{fig:non-complementable-gt-1}}
\end{figure*}
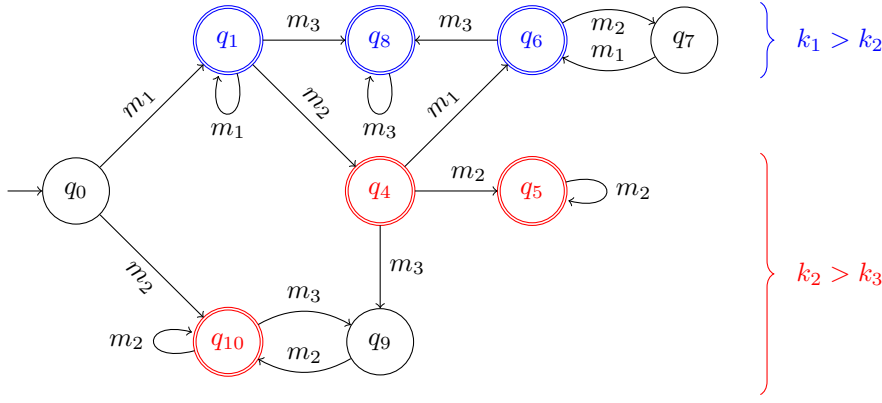 

\begin{figure}
    \begin{tikzpicture}
    \draw[->] (0,0) node [above] {p} -- (0,-6);
    \draw[->] (1,0) node [above] {q} -- (1,-6);
    \draw[->] (2,0) node [above] {r} -- (2,-6);
    \draw[->] (3,0) node [above] {s} -- (3,-6);

    \draw[->] (0,-.5) -- node [above] {$m_1$} (1,-.5);
    \draw[->] (0,-1) -- node [above] {$m_1$} (1,-1);
    \node at (.5,-1.2) {$\vdots$};
    \draw[->] (0,-2.1) -- node [above] {$m_1$} (1,-2.1);

    \draw[->] (2,-1.5) -- node [above] {$m_2$} (3,-1.5);
    \draw[->] (2,-2) -- node [above] {$m_2$} (3,-2);
    \node at (2.5,-2.2) {$\vdots$};
    \draw[->] (2,-3.1) -- node [above] {$m_2$} (3,-3.1);

    \draw[->] (0,-3.5) -- node [above] {$m_3$} (1,-3.5);
    \draw[->] (0,-4) -- node [above] {$m_3$} (1,-4);
    \node at (.5,-4.2) {$\vdots$};
    \draw[->] (0,-5.1) -- node [above] {$m_3$} (1,-5.1);

    % on ajoute des accolades
    \draw [decorate,decoration={brace,amplitude=5pt,mirror}] (-.3,-.5) -- node [midway, left=1em] {$k_1$} (-.3,-2.1) ;
    \draw [decorate,decoration={brace,amplitude=5pt,mirror}] (-.3,-3.5) -- node [midway, left=1em] {$k_3$} (-.3,-5.1) ;
    \draw [decorate,decoration={brace,amplitude=5pt}] (3.3,-1.5) -- node [midway, right=1em] {$k_2$} (3.3,-3.1) ;

\end{tikzpicture}
    \caption{The shape of the MSCs in $\msclanguageofcc{\gt_0}{}$ for the global type $\gt_0$.
    \label{fig:msc-regular}}
\end{figure}

\subsection{Complementation By Duality}

We write $\dualdfaof{\gt}$ for the dual DFA of 
a deterministic
global type $\gt$, where accepting states and non-accepting ones are swapped (possibly completing first $\gt$ with a sink state).
It follows from the definition of $\existentialmsclanguageof{\gt}$ 
and $\universalmsclanguageof{\gt}$ that
$$
\existentialmsclanguageof{\dualdfaof{\gt}}=\mscsetofmodel{\synchmodel}\setminus\universalmsclanguageof{\gt}.
$$
In general $\dualdfaof{\gt}$ is not a complement of $\gt$; still  
duality can be used to obtain a complement of $\gt$ 
in a few cases. The first and most obvious case is the following.
\begin{proposition}\label{prop:commutation-closed-implies-complementable}
    If $\gt$ is deterministic and commutation-closed, then $\dualdfaof{\gt}$ is a complement of $\gt$.  
\end{proposition}

As observed before, all global types with less than three participants are commutation closed, 
hence from Lemma~\ref{prop:commutation-closed-implies-complementable}
we have the following immediate result.

\begin{corollary}\label{coro:3-participants-complementable}
    If $\cardinalof{\Procs}\leq 3$, then $\dualdfaof{\gt}$ is a complement of $\gt$.  
\end{corollary}

Now, even if a global type $\gt$ is not commutation-closed, 
it may be "enlarged" as the commutation-closed global type $\productof{\projectionof{\gt}}$,
a process sometimes called Cartesian abstraction~\cite{cartesian-abstraction}.
This commutation-closed global type can later be complemented by duality to an under-approximation of a complement of $\gt$. When $\gt$ is deadlock-free realisable in $\synchmodel$, this under-approximation is exact.

\begin{theorem}\label{thm:realisable-complementable}
    If a global type $\gt$ is deadlock-free realisable in $\synchmodel$, 
    then
    $\dualdfaof{\detof{\productof{\projectionof{\gt}}}}$
    is a complement of $\gt$.
\end{theorem}
\begin{proof}
    Let $\gt$ be a global type that is  deadlock-free realisable in $\synchmodel$.
    By Condition~(CC) of Definition~\ref{def:realisability}, 
    $$
    \executionsof{\projectionof{\gt}}{\synchmodel} = \executionsof{\gt}{\synchmodel}.
    $$
    If two sets of executions are equal, the corresponding sets of their MSCs are also equal, 
    and thus 
    $$
    \msclanguageof{\projectionof{\gt}}{\synchmodel}=\msclanguageof{\gt}{\synchmodel}.
    $$
    By Proposition~\ref{prop:product-is-commutation-closed}, the synchronous product $\productof{\projectionof{\gt}}$
    is a commutation-closed global type whose runs are exactly the synchronous executions of the projected system, so 
    $$\msclanguageofcc{\productof{\projectionof{\gt}}}=\msclanguageof{\projectionof{\gt}}{\synchmodel}=\existentialmsclanguageof{\gt}.$$

    Finally, by Proposition~\ref{prop:commutation-closed-implies-complementable},
    $\dualdfaof{\detof{\productof{\projectionof{\gt}}}}$ is a complement of $\productof{\projectionof{\gt}}$, thus of $\gt$.
\end{proof}

\begin{remark}
    In terms of complexity, complementation by duality is
    linear in the number of states of the global type $\gt$ (possibly after adding a sink state). However, the determined Cartesian abstraction 
    $\detof{\productof{\projectionof{\gt}}}$ involves in the worst case a doubly exponential blowup. 
\end{remark}

\subsection{Complementation by Renunciation}

In this section, we introduce another complementation procedure  for
global types with sender-driven choices, and more generally presenting a form of determinism in commutations; this construction is also linear, although it does not preserve determinism. We first recall the definition of sender-driven choice and commutation-determinism, then define our complementation procedure, and finally establish its correctness.

For a state $s$ of $\gt$, let $\choicesof{\gt}{s}$ be the set of arrows labeling outgoing transitions of $s$ in $\gt$.

\begin{definition}[Sender-Driven~\cite{stutz-phd}]\label{def:sender-driven}
    A global type $\gt$ is \emph{sender-driven} if it is deterministic and for every state $s$ of $\gt$,
    $$
    \choicesof{\gt}{s} = \{\marrow{p}{q_i}{m_i}\mid i=1,\ldots,n\}
    $$
    for some process $p$ and some processes $q_i$ and messages $m_i$.
\end{definition}

Stutz et al. introduced this assumption on global types in order to prove that $\ppmodel$-implementability (a notion close to realisability) is decidable~\cite{DBLP:conf/cav/LiSWZ23}. We slightly generalise this condition as the complementation procedure we are about to present also works in this larger setting.

\begin{definition}[Commutation-determinismic]
    A global type $\gt$ is commutation-deterministic if it is deterministic and for every state $s$ of $\gt$,
    for every two arrows $a,b\in\choicesof{\gt}{s}$, $a$ and $b$ do not commute.
\end{definition}

\begin{figure*}[ht]
        \begin{tikzpicture}
        \begin{scope}
            \node at (-1,0) {$\gt~=$};
            \node[state,initial,initial text={},initial distance=3mm] (s0) at (1,0) {$s_0$};
            \node[state] (s1) at (3.5,1) {$s_1$};
            \node[state,accepting] (s2) at (3.5,0)  {$s_2$};
            \node[state,accepting] (s3) at (7,1) {$s_3$};
            \draw[->] (s0) -- node[above,sloped] {$\gtlabel{p}{q}{m_1}$} (s1);
            \draw[->] (s0) -- node[below] {$\gtlabel{p}{q'}{m_2}$} (s2);
            \draw[->] (s1) -- node[above] {$\gtlabel{r}{r'}{m_3}$} (s3);
        \end{scope}
        \begin{scope}[xshift=8.5cm, yshift=1cm, scale=.5]
            \node at (0,0) {$M_1~=$};
            \draw[->] (1,1) node [above] {q}  -- (1,-1);
            \draw[->] (2,1) node [above] {p}  -- (2,-1);
            \draw[->] (3,1) node [above] {q'} -- (3,-1);
            \draw[->] (4,1) node [above] {r}  -- (4,-1);
            \draw[->] (5,1) node [above] {r'} -- (5,-1);
            \draw[->] (2,0) -- node [above] {$a_1$} (1,0);
            \draw[->] (4,0) -- node [above] {$a_3$} (5,0);
        \end{scope}
        \begin{scope}[xshift=9.5cm, yshift=-1cm, scale=.5]
            \node at (0,0) {$M_2~=$};
            \draw[->] (1,1) node [above] {q}  -- (1,-1);
            \draw[->] (2,1) node [above] {p}  -- (2,-1);
            \draw[->] (3,1) node [above] {q'} -- (3,-1);
            \draw[->] (4,1) node [above] {r}  -- (4,-1);
            \draw[->] (5,1) node [above] {r'} -- (5,-1);
            \draw[->] (2,0) -- node [above] {$a_2$} (3,0);
        \end{scope}
        \begin{scope}[yshift=-6cm]
            \node at (-1,0) {$\renun{\gt}~=$};
            \node[state,initial,initial text={},initial distance=3mm, accepting] (s0) at (1,0) {$s_0$};
            \node[state,accepting] (s1) at (1,-3) {$s_1$};
            \node[state,accepting] (s0bar) at (1,3)  {$\comp{s_0}$};
            \node[state,accepting] (s1bar) at (-2,-3) {$\comp{s_1}$};
            \node[state] (s0a1) at (5,0) {$s_0,a_1$};
            \node[state] (s0a2) at (5,3) {$s_0,a_2$};
            \node[state] (s0bara1) at (9,0) {$\comp{s_0},a_1$};
            \node[state] (s0bara2) at (9,3) {$\comp{s_0},a_2$};
            \node[state,accepting] (sacc) at (11,0) {$s_{\mathsf{acc}}$};
            \node[state] (s1a3) at (5,-3) {$s_1,a_3$};
            \node[state] (s1bara3) at (9,-3) {$\comp{s_1},a_3$};
            \draw[->] (s0) -- node[left] {$a_1$} (s1);
            \draw[->] (s0) -- node[left] {$\neg a_1\wedge \neg a_2$} (s0bar);
            \draw[->] (s0bar) edge [loop above] node[above] {$\neg a_1\wedge \neg a_2$} (s0bar);
            \draw[->] (s0) edge [bend right] node[above,sloped] {$\neg a_1\wedge \neg a_2\wedge\conflictingarrows{a_1}$} (s0bara1);
            \draw[->] (s0) -- node[above,sloped] {$\neg a_1\wedge \neg a_2\wedge\neg\conflictingarrows{a_1}$} (s0a1);
            \draw[->] (s0a1) edge [loop above] node[above] {$\neg a_1\wedge \neg a_2\wedge\neg\conflictingarrows{a_1}$} (s0a1);
            \draw[->] (s0a1) -- node[above,sloped] {$\neg a_1\wedge \neg a_2\wedge\conflictingarrows{a_1}$} (s0bara1);
            \draw[->] (s0bara1) edge [loop above] node[above] {$\neg a_1\wedge\neg a_2$} (s0bara1);
            \draw[->] (s0bara1) -- node[above,sloped] {$a_1$} (sacc);
            \draw[->] (s0) edge [bend left=9] node[above,sloped,pos=0.6] {$\neg a_1\wedge \neg a_2\wedge\conflictingarrows{a_2}$} (s0bara2);
            \draw[->] (s0) -- node[above,sloped] {$\neg a_1\wedge \neg a_2\wedge\neg\conflictingarrows{a_2}$} (s0a2);
            \draw[->] (s0) -- node[above,sloped] {$\neg a_1\wedge \neg a_2\wedge\neg\conflictingarrows{a_2}$} (s0a2);
            \draw[->] (s0a2) edge [loop above] node[above] {$\neg a_1\wedge \neg a_2\wedge\neg\conflictingarrows{a_2}$} (s0a2);
            \draw[->] (s0a2) -- node[above,sloped] {$\neg a_1\wedge \neg a_2\wedge\conflictingarrows{a_2}$} (s0bara2);
            \draw[->] (s0bara2) edge [loop above] node[above] {$\neg a_1\wedge\neg a_2$} (s0bara2);
            \draw[->] (s0bara2) -- node[above,sloped] {$a_2$} (sacc);
            \draw[->] (s1) -- node [above] {$\neg a_3$}(s1bar);
            \draw[->] (s1bar) edge[loop above] node[above] {$\neg a_3$} (s1bar);
            \draw[->] (s1) edge [bend right] node[above,sloped] {$\neg a_3\wedge \conflictingarrows{a_3}$} (s1bara3);
            \draw[->] (s1) -- node[above,sloped] {$\neg a_3\wedge\neg\conflictingarrows{a_3}$} (s1a3);
            \draw[->] (s1a3) edge [loop above] node[above] {$\neg a_3\wedge\neg\conflictingarrows{a_3}$} (s1a3);
            \draw[->] (s1a3) -- node[above,sloped] {$\neg a_3\wedge \conflictingarrows{a_3}$} (s1bara3);
            \draw[->] (s1bara3) edge [loop above] node[above] {$\neg a_3$} (s1bara3);
            \draw[->] (s1bara3) -- node[above,sloped] {$a_3$} (sacc);

            \draw[->] (sacc) edge [loop right] node[right] {$\Arrows$} (sacc);

        \end{scope}

    \end{tikzpicture}
    \caption{A sender-driven global type $\gt$, with $\existentialmsclanguageof{\gt}=\{M_1,M_2\}$, and its
            complement $\renun{\gt}$.
            For better readability, the states $s_2$ and $s_3$ (that are sink states in $\renun{\gt}$) have been pruned.
            $\conflictingarrows{a}$ denotes the set of arrows that do not commute with $a$.
            }\label{fig:example-of-sender-driven-gt-complementation}
\end{figure*}

Let us now formally define the complementation procedure for commutation-deterministic global types.

\begin{definition}[Complementation by renunciation]\label{def:complement-of-a-sender-driven-gt}
    Let $\gt$ be a commutation-deterministic global type. Let $S$ be the set of states of $\gt$, $\comp{S}\eqdef\{\comp{s}\mid s\in S\}$,
    and $S'=S\uplus\comp{S}$.
   $\renun{\gt}$ is the global type with set of states $S'\cup S'\times\Arrows \cup\{s_{\mathsf{acc}}\}$ defined as follows:
    \begin{itemize}
        \item the initial state of $\renun{\gt}$ is the initial state of $\gt$;
        \item let $F\subseteq S$ be the set of accepting states of $\gt$; the set of accepting states of 
        $\renun{\gt}$ is $(S\setminus F)\cup \comp{S}\cup\{s_{\mathsf{acc}}\}$;
        \item for any states $s,s'\in S$ of $\gt$, for any arrow $a\in \Arrows$,
        if $(s,a,s')$ is a transition in $\gt$, then $(s,a,s')$ 
        is a transition in $\renun{\gt}$;
        \item for any state $s$ of $\gt$, 
        for any arrow $a\not\in\choicesof{\gt}{s}$,
        $(s,a,\comp{s})$ (resp. $(\comp{s},a,\comp{s})$) 
        is a transition in $\renun{\gt}$;
        \item for any state $s$ of $\gt$,
        for any arrow $a\in\choicesof{\gt}{s}$,
        for any arrow $b\not\in\choicesof{\gt}{s}$ that commutes with $a$,
        $\big(s,b,(s,a)\big)$ 
        (resp. $\big((s,a),b,(s,a)\big)$) 
        is a transition in $\renun{\gt}$
        \item for any state $s$ of $\gt$,
        for any arrow $a\in\choicesof{\gt}{s}$,
        for any arrow $b\not\in\choicesof{\gt}{s}$ that does not commute with $a$,
        $\big(s,b,(\comp{s},a)\big)$ 
        (resp. $\big((s,a),b,(\comp{s},a)\big)$) 
        is a transition in $\renun{\gt}$
        \item for any state $s$ of $\gt$,
        for any arrow $a\in\choicesof{\gt}{s}$,
        for any arrow $b\not\in\choicesof{\gt}{s}$,
        $\big((\comp{s},a),b,(\comp{s},a)\big)$
        (resp. $\big((\comp{s},a),a,s_{\mathsf{acc}}\big)$)
        is a transition in $\renun{\gt}$;
        \item for any arrow $a$, $(s_{\mathsf{acc}},a,s_{\mathsf{acc}})$ is a transition in $\renun{\gt}$.
    \end{itemize}
\end{definition}

\begin{example}\label{ex:example-of-sender-driven-gt-complementation}
    Fig.~\ref{fig:example-of-sender-driven-gt-complementation} depicts a sender-driven global type
    and the complement computed according to Definition~\ref{def:complement-of-a-sender-driven-gt}.
\end{example}

Note that the number of states of $\renun{\gt}$ is linear in the number of states of $\gt$ (assuming the alphabet $\Arrows$ is fixed); it can also be observed that, using Boolean
expressions to label transitions, the size of $\renun{\gt}$ can also be kept linear in the size of $\gt$.

\begin{example}
    Let $\gt,\renun{\gt}, M_1, M_2$ be the global types and MSCs depicted in Fig.~\ref{fig:example-of-sender-driven-gt-complementation}.
    It is easy to verify that $M_2\not\in\existentialmsclanguageof{\renun{\gt}}$, because the only sequence of arrows
    $w_2=a_2$ such that $M_2=\mscof{w_2}$ is not in $\languageofnfa{\renun{\gt}}$.
    It can also be checked that $M_1\not\in\existentialmsclanguageof{\renun{\gt}}$, because the only two  sequences of arrows $w_{1,1}=a_1a_3$
    and $w_{1,2}=a_3a_1$ such that $\mscof{w_{1,1}}=\mscof{w_{1,2}}=M_1$ are not in $\languageofnfa{\renun{\gt}}$.
\end{example}

Intuitively, $\renun{\gt}$ describes a MSC $M$ that starts with a prefix that is in $\existentialmsclanguageof{\gt}$ up to a certain state $s$, at which point a "renunciation" of the choices of $\gt$ at $s$ occurs. 
There are two kinds of renunciation: definitive renunciation means that none of the arrows available in the choices will ever occur later.
Provisory renunciation means that at least one arrow in the choices will occur later, but not immediately: if $a$ is the first
arrow of the choice that occurs later, another arrow $b$ occurs before $a$ such that $b$ does not commute with $a$ and is not in the choices that
have been renunciated. The first kind of renunciation corresponds to moving to a state of the form $\comp{s}$, while the second kind corresponds to moving to a state of the form $(s,a)$ or $(\comp{s},a)$.

\begin{example}
    Consider the MSC $M_3=\mscof{a_1a_2'a_3}$, with $a_2'=\gtlabel{p}{q'}{m_2'}$ for some $m_2'\neq m_2$. Then $M_3$ is accepted by $\renun{\gt}$ because $w_{3,1}=a_1a_2'a_3$ is accepted in $\comp{s_1}$. However, the other sequence of arrows $w_{3,2}=a_3a_1a_2'$
    such that $\mscof{w_{3,2}}=M_3$ is not in $\languageofnfa{\renun{\gt}}$. More generally, $\renun{\gt}$ is \emph{not} commutation-closed, and for any MSC $M\not\in\existentialmsclanguageof{\gt}$, the
    sequence of arrows $w$ such that $M=\mscof{w}$ and $w\in\languageofnfa{\renun{\gt}}$ should be carefully constructed. The strategy
    consists in picking a sequence of arrows $w=w_1\cdot w_2$
    with the longest possible $w_1$ 
    in $\prefixclosureof{\languageofnfa{\gt}}$, i.e. "renunciating" to a choice as late as possible.
\end{example}

In order to formalise this intuition, we introduce a few notions.
Given an MSC $M$ and a state $s$ of a commutation-deterministic global type $\gt$,
we write $\nextarrow{M}{s}$ for the first arrow of $\choicesof{\gt}{s}$ that occurs in a sequence of arrows $w$
such that $\mscof{w}=M$. Note that $\nextarrow{M}{s}$ is not defined if $M$ contains no arrow of 
$\choicesof{\gt}{s}$; note also that $\nextarrow{M}{s}$, when defined, does not depend on the choice of the sequence of arrows 
$w$ such that $\mscof{w}=M$ (otherwise two arrows of $\choicesof{\gt}{s}$ would commute, contradicting the commutation-determinism of $\gt$).

\begin{example}
    Consider the MSC $M_4=\mscof{a_1a_2a_3}=\mscof{a_3a_1a_2}$, and the initial state $s_0$ of $\gt$ in Fig.~\ref{fig:example-of-sender-driven-gt-complementation}.
    Then $\nextarrow{M_4}{s_0}=a_1$.
\end{example}

We also define $\nextmsc{M}{s}$ as $\mscof{w}$ for some $w$ such
that $M=\mscof{\nextarrow{M}{s}w}$; intuitively, we "remove" the arrow $\nextarrow{M}{s}$ from $M$, provided it is not blocked by 
a previous non-commuting arrow. Note that $\nextmsc{M}{s}$ is not defined if 
in all sequences of arrows $w$ such that $\mscof{w}=M$, $\nextarrow{M}{s}$ never occurs first.

\begin{example}
    Consider the MSC $M_4=\mscof{a_1a_2a_3}$, and the state $s_0$ of $\gt$ in Fig.~\ref{fig:example-of-sender-driven-gt-complementation}.
    Then $\nextmsc{M_4}{s_0}=\mscof{a_2a_3}$.
    On the other hand, for the MSC $M_5=\mscof{a_4a_1a_2}$, with $a_4=\gtlabel{q}{q'}{m_4}$, $\nextmsc{M_5}{s_0}$
    is undefined, because $a_1$ is "blocked" by $a_4$.
\end{example}

\begin{lemma}\label{lem:nextarrow-and-nextmsc}
    Let $\gt$ be a commutation-deterministic global type with initial state $s_0$, and $M$ a non-empty MSC.
    Then $M\not\in\existentialmsclanguageof{\gt}$ if and only if one of the following holds:
    \begin{itemize}
        \item $\nextarrow{M}{s_0}$ is undefined, or 
        \item $\nextmsc{M}{s_0}$ is undefined, or
        \item $a=\nextarrow{M}{s_0}$, $s$ is the $a$-successor of $s_0$ in $\gt$,
        $M'=\nextmsc{M}{s_0}$ and 
        $M'\not\in\existentialmsclanguageof{\gt_s}$,
        where $\gt_s$ is the global type obtained from $\gt$ by setting the initial state to $s$.
    \end{itemize}
\end{lemma}

\begin{proof}[Proof (sketch)]
    If $\nextarrow{M}{s_0}$ is undefined, then $M$ contains no arrow of $\choicesof{\gt}{s_0}$, 
    so for every sequence of arrows $w$ such that $\mscof{w}=M$, $w$ does not start with
    an arrow of $\choicesof{\gt}{s_0}$, therefore $w\not\in\languageofnfa{\gt}$.
    Similarly, when $\nextmsc{M}{s_0}$ is undefined, every sequence of arrows $w$ such that $\mscof{w}=M$ does not start with
    an arrow of $\choicesof{\gt}{s_0}$, and again $w\not\in\languageofnfa{\gt}$.
    In the third case, assume by contradiction that $M\in\existentialmsclanguageof{\gt}$,
    i.e. there is a sequence of arrows $w$ such that $\mscof{w}=M$ and $w\in\languageofnfa{\gt}$.
    Then $w=aw'$ for some $a=\nextarrow{M}{s_0}$ and $w'\in\languageofnfa{\gt,s}$, with $\nextmsc{M}{s_0}=\mscof{w'}$. It follows that     
    $\nextmsc{M}{s_0}\in\existentialmsclanguageof{\gt_s}$, and the contradiction.
\end{proof}

\begin{restatable}{theorem}{thmsenderdrivenchoicecomplementation}\label{thm:sender-driven-choice-complementation}
    Assume $\gt$ is a commutation-deterministic global type,
    and let $\renun{\gt}$ denote the global type defined as in Def~\ref{def:complement-of-a-sender-driven-gt}.
    Then $\renun{\gt}$ is a complement of $\gt$.
\end{restatable}

\begin{proof}[Proof (sketch)]
    It is routine to check that an MSC $M$ is accepted by $\renun{\gt}$ 
    starting from a state $\comp{s}$ if and only if $\nextarrow{M}{s}$ is undefined. It is also routine to check that an MSC $M$ is accepted by $\renun{\gt}$
    starting from a state $(s,a)$ if $a=\nextarrow{M}{s}$ and $\nextmsc{M}{s}$ is undefined.
    It is then also routine, thanks to Lemma~\ref{lem:nextarrow-and-nextmsc}, to show by induction on the size of an MSC $M$ that
    $M\not\in\existentialmsclanguageof{\gt_s}$ if and only if $M\in\existentialmsclanguageof{\renun{\gt_s}}$.
\end{proof}

	\section{Decidability of Realisability}\label{sec:decidability}
	
% !TEX root =  ../mainPPDP.tex
%!TEX spellcheck = en_GB

In this section, we show that the realisability problem for complementable global types,
given with an explicit complement, is decidable in both the synchronous and the p2p model.
We first show the result for the synchronous model, which is simpler,
and then extend it to the p2p model using a characterisation of realisability in the p2p model.

\subsection{Realisability in the synchronous model}

\begin{theorem}
    \label{thm:decidability-of-implementability-in-synch}
    Given a global types $\gt$ and a complement $\comp{\gt}$ of $\gt$,
    it is decidable in PSPACE complexity whether  $\gt$ is deadlock-free realisable in $\synchmodel$.    
%  The  problem is in PSPACE:
  %  \begin{itemize}
    %    \item Input: a global type $\gt$ and a complement $\gt'$ of $\gt$.
      %  \item Question: is ?
    %\end{itemize}
\end{theorem}

\begin{proof}(Sketch)
    Condition~(CC) of Definition~\ref{def:realisability}
    is equivalent to $$\msclanguageofcc{\productof{\projectionof{\gt}}} \cap \existentialmsclanguageof{\comp{\gt}} = \emptyset.$$
    By Lemma~\ref{lem:product-of-gt}, this reduces to checking the emptiness of $\languageofnfa{\productof{\projectionof{\gt}}\otimes\comp{\gt}}$.
    The non-emptiness of the language of an NFA
    can be checked in non-deterministic logarithmic space in the size of the NFA,
    which yields a PSPACE algorithm in that case, as the exponential size global type  $\productof{\projectionof{\gt}}$
    can be lazily constructed while guessing a path to an accepting state.
    
    Assuming that Condition~(CC) of Definition~\ref{def:realisability} holds,
    the second condition of Definition~\ref{def:realisability}  is 
    equivalent to whether all control states 
    of $\productof{\projectionof{\gt}}$ that are reachable from the initial state
    can reach a final state, which again, for the same reason, can be checked in PSPACE.
\end{proof}

\subsection{Realisability in the $\ppmodel$ model}

We now consider the realisability problem in the $\ppmodel$ model.
Note that relaxing the communication model may both remove some deadlock situations,
but also introduce some new ones (because new configurations may be reachable).
So realisability in the $\ppmodel$ model 
and realisability in the $\synchmodel$ model could be expected to be two different properties,
none of which implies the other. Due to the the fact that global types enforce a form
of synchronous-like behaviour, however, it turns out that every global type that is 
realisable in the $\ppmodel$ model is also realisable in the $\synchmodel$ model.

\begin{theorem}\label{thm:pp-realizability-implies-synch-realizability}
    If $\gt$ is deadlock-free realisable in $\ppmodel$, then $\gt$ is deadlock-free realisable in $\synchmodel$.
\end{theorem}
\iflong
\begin{proof}
Assume that $\gt$ is $\ppmodel$-realisable. We show that $\gt$ is $\synchmodel$-realisable by verifying the 
two conditions of Definition~\ref{def:realisability}.   
\begin{itemize}
\item 
    Since $\executionsofmodel{\synchmodel}\subseteq\executionsofmodel{\ppmodel}$, we have
    $$\mscsofcfsms{\projectionof{\gt}}{\synchmodel}\subseteq\mscsofcfsms{\projectionof{\gt}}{\ppmodel}.$$
    By Proposition~\ref{prop:msc-version-of-cond1-of-realizability} and the hypothesis that 
    $\gt$ is $\ppmodel$-realisable, we have that
    $$\mscsofcfsms{\projectionof{\gt}}{\ppmodel} \subseteq \existentialmsclanguageof{\gt}.
    $$
    Therefore, we have
    $$\mscsofcfsms{\projectionof{\gt}}{\synchmodel} \subseteq \existentialmsclanguageof{\gt}$$
    and by Proposition~\ref{prop:msc-version-of-cond1-of-realizability}, $\gt$ satisfies condition~1 of Definition~\ref{def:realisability}
    for being deadlock-free realisable in $\synchmodel$.

\item By Proposition~\ref{prop:deadlock-free-as-a-property-on-mscs-for-p2p-and-synch} and condition~1, the condition~2 of
    Definition~\ref{def:realisability} is equivalent to
    $$
        \mscsofcfsms{\projectionof{\acceptcompletion{\gt}}}{\acommunicationmodel}\subseteq\prefixclosureof{\existentialmsclanguageof{\gt}}
    $$
    when $\acommunicationmodel$ is either $\ppmodel$ or $\synchmodel$.
    We can therefore assume that 
    $$
        \mscsofcfsms{\projectionof{\acceptcompletion{\gt}}}{\ppmodel}\subseteq\prefixclosureof{\existentialmsclanguageof{\gt}}
    $$
    and we show that 
    $$
        \mscsofcfsms{\projectionof{\acceptcompletion{\gt}}}{\synchmodel}\subseteq\prefixclosureof{\existentialmsclanguageof{\gt}}
    $$
    This follows from the fact that
    $$
        \mscsofcfsms{\projectionof{\acceptcompletion{\gt}}}{\synchmodel}\subseteq\mscsofcfsms{\projectionof{\acceptcompletion{\gt}}}{\ppmodel}
    $$
    because $\executionsofmodel{\synchmodel}\subseteq\executionsofmodel{\ppmodel}$.
\end{itemize}
\end{proof}

\fi

We now show a "converse" of Theorem~\ref{thm:pp-realizability-implies-synch-realizability},
which precisely states which additional conditions are needed to ensure that a global type that
is deadlock-free realisable in the $\synchmodel$ model is also deadlock-free realisable in the $\ppmodel$ model.

\begin{restatable}[Reduction to realisability in $\synchmodel$]{theorem}{maintheoremrealisability}
    \label{thm:main-theorem-realisability}
        A global type $\gt$ is  deadlock-free realisable in $\ppmodel$
        if and only if the following four conditions hold:
        \begin{enumerate}
            \item $\msclanguageof{\projectionof{\gt}}{\ppmodel}\subseteq\prefixclosureof{\mscsetofmodel{\synchmodel}}$
            \item $\projectionof{\gt}$ is orphan-free in $\ppmodel$, 
            \item $\mscsofcfsms{\projectionof{\acceptcompletion{\gt}}}{\ppmodel}\subseteq\prefixclosureof{\mscsofcfsms{\projectionof{\gt}}{\ppmodel}}$, and 
            \item $\gt$ is deadlock-free realisable in $\synchmodel$.
        \end{enumerate}  
\end{restatable}
\iflong
\begin{proof}
We show both sides in turn:
\begin{description}
    \item[$\Rightarrow$]
    We first show that the four conditions are necessary.
    Let $\gt$ be a $\ppmodel$-realisable global type.
    \begin{enumerate}
        \item        
            By Condition~2 of Definition~\ref{def:realisability},
            $\projectionof{\gt}$ is deadlock-free in $\ppmodel$.
            By Proposition~\ref{prop:deadlock-free-as-a-property-on-mscs-for-p2p-and-synch},
            \[
                \ \ \ \ \ \ \ \mscsofcfsms{\projectionof{\acceptcompletion{\gt}}}{\ppmodel}\subseteq\prefixclosureof{\mscsofcfsms{\projectionof{\gt}}{\ppmodel}}.
            \]
            By Condition~1 of Definition~\ref{def:realisability},
            \[
                \mscsofcfsms{\projectionof{\gt}}{\ppmodel} = \existentialmsclanguageof{\gt} \subseteq \mscsetofmodel{\synchmodel}.
            \]
            Therefore $\mscsofcfsms{\projectionof{\acceptcompletion{\gt}}}{\ppmodel}\subseteq\prefixclosureof{\mscsetofmodel{\synchmodel}}$.
   
        \item Let $e\in\executionsof{\projectionof{\gt}}{\ppmodel}$;
        we show that $e$ is orphan-free.
        By Condition~1 of Definition~\ref{def:realisability},
        $\mscof{e}\in\existentialmsclanguageof{\gt}$, therefore
        $\mscof{e}$ is synchronous, and in particular orphan-free.
       
        \item $\mscsofcfsms{\projectionof{\acceptcompletion{\gt}}}{\ppmodel}\subseteq\prefixclosureof{\mscsofcfsms{\projectionof{\gt}}{\ppmodel}}$ 
            by Condition~2 of Definition~\ref{def:realisability} and Proposition~\ref{prop:deadlock-free-as-a-property-on-mscs-for-p2p-and-synch}.

        \item  Follows from Lemma~\ref{thm:pp-realizability-implies-synch-realizability} 
    \end{enumerate}
\item[$\Leftarrow$]    Next we  show that the four conditions are sufficient.
    Let $\gt$ be a global type that verifies the four conditions.
%    \etienne{Ca marche pas! contre-exemple: $\gt = p->q:a + r->q:b$!}
    We prove that $\gt$ is $\ppmodel$-realisable as it satisfies the two conditions of Definition~\ref{def:realisability}.
    \begin{description}
        \item[Condition 1]
          By Proposition~\ref{prop:msc-version-of-cond1-of-realizability},
          we show that 
          \[
          \mscsofcfsms{\projectionof{\gt}}{\ppmodel} \subseteq \existentialmsclanguageof{\gt}.
          \]
          By Condition~1 of Theorem~\ref{thm:main-theorem-realisability},
          \[
            \mscsofcfsms{\projectionof{\gt}}{\ppmodel} = \mscsofcfsms{\projectionof{\gt}}{\synchmodel}
          \]
          and by Condition~4 of Theorem~\ref{thm:main-theorem-realisability} and Proposition~\ref{prop:msc-version-of-cond1-of-realizability},
          we have that  
          \[
            \mscsofcfsms{\projectionof{\gt}}{\synchmodel} \subseteq \existentialmsclanguageof{\gt}
          \]
          which concludes the proof of Condition~1.
          \item[Condition 2]
        follows from Proposition~\ref{prop:deadlock-free-as-a-property-on-mscs-for-p2p-and-synch}
        and Condition~3 of Theorem~\ref{thm:main-theorem-realisability}.
    \end{description}
\end{description}
\end{proof}
    
\fi

A system $\system$ with 
$\msclanguageof{\system}{\ppmodel}\subseteq
\prefixclosureof{\mscsetofmodel{\synchmodel}}$
is called \emph{realisable with synchronous communications} (RSC) 
after Germerie~\cite{germerie-phd}. Germerie
showed that the RSC property is decidable and that
whether a RSC system may reach a regular set of configurations is
also decidable. As a consequence, conditions~1 and 2 of
Theorem~\ref{thm:main-theorem-realisability} are decidable.
Assuming that $\projectionof{\gt}$ is RSC, 
$\mscsofcfsms{\projectionof{\acceptcompletion{\gt}}}{\ppmodel}$ can be coded as an effective regular set of
executions, as well as $\prefixclosureof{\mscsofcfsms{\projectionof{\gt}}{\ppmodel}}$,
and Condition~3 of Theorem~\ref{thm:main-theorem-realisability} reduces to checking the inclusion of
two NFAs, which can be checked in PSPACE using a lazy construction of the NFAs.

From these results we can conclude that it is decidable whether a given complementable global type,
together with its complement, is deadlock-free realisable in the $\ppmodel$ model.

\begin{theorem}\label{thm:decidability-of-implementability-in-p2p}   
   Given a complementable global type $\gt$ and a complement $\comp{\gt}$, 
   it is decidable whether $\gt$ is deadlock-free realisable in $\ppmodel$.
\end{theorem}

	\section{Concluding remarks} \label{sec:concl}	
	  % !TEX root =  ../mainPPDP.tex
%!TEX spellcheck = en_GB

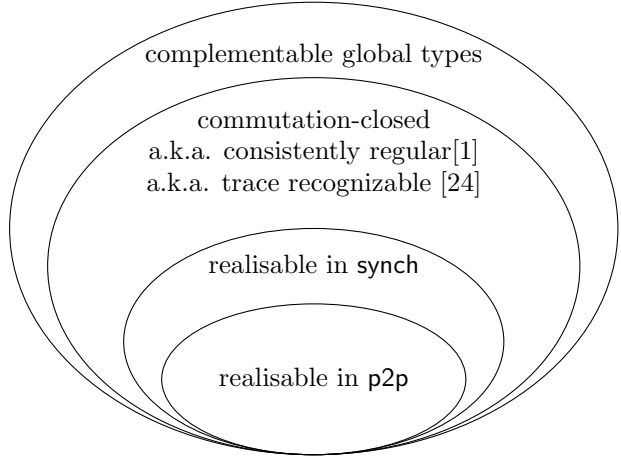
\begin{figure}
	\centering
	\begin{tikzpicture}
    \node[draw, ellipse, minimum width=8cm, minimum height=6cm, align=center] at (0,4.5) {};
    \node[draw, ellipse, minimum width=7cm, minimum height=5cm, align=center] at (0,4) {};
    \node[draw, ellipse, minimum width=5cm, minimum height=3cm, align=center] at (0,3) {};
    \node[draw, ellipse, minimum width=4cm, minimum height=2cm, align=center] at (0,2.5) {};
    \node at (0,6.8) {complementable global types};
    \node at (0,5.5) {\begin{tabular}{c}commutation-closed\\ a.k.a. consistently regular\cite{AALBERSBERG19881} \\ a.k.a. trace recognizable~\cite{DBLP:conf/ershov/Zielonka89}\end{tabular}};
    \node at (0,4) {realisable in $\synchmodel$};
    \node at (0,2.5) {realisable in $\ppmodel$};
\end{tikzpicture}
	\caption{Hierarchy of semantic properties of global types}\label{fig:gt-hierarchy}
\end{figure}

In this paper, we investigated the realisability and complementability of MPSTs across two canonical communication models: synchronous and asynchronous ($\ppmodel$). 
We reduce the task of checking realisability in the complex asynchronous setting to the simpler synchronous case, as long as the initial system satisfies key properties such as deadlock-freedom, absence of orphan messages, and a form of "synchronisability" (RSC property~\cite{germerie-phd}). Said differently, we showed that global types that are realisable in the $\ppmodel$ model form a recursive subclass of the ones realisable in $\synchmodel$. 
We also showed that this latter class itself is a subclass of the class of global types that are equivalent (in terms of MSC semantics) to a commutation-closed global type,
and these form a subclass of the complementable ones (see Figure~\ref{fig:gt-hierarchy}).
Among others, it suggests that any framework enforcing "by construction" the realisability of 
global types "hides" a complementation procedure for the global types embraced by this framework.

We explored the complementability of global types, looking for complementation procedures.
We first observed that the complementation of a global type is not always possible, but we 
proposed several complementation procedures for various subclasses of global types, among which the one of sender-driven global types. For each of these subclasses, we derive a decision procedure for realisability in $\synchmodel$ with PSPACE complexity, and for 
realisability in $\ppmodel$ with at most EXPSPACE complexity. Figure~\ref{fig:complementation-procedures-summary} summarises these results.

\begin{figure*}[ht]
\begin{center}
{\small
\begin{tabular}{|c|c|c|c|}
\hline
\textbf{Class of global types} & \textbf{Complementation procedure} & \begin{tabular}{c}\textbf{Size of}\\ \textbf{complement} \end{tabular} & \textbf{$\synchmodel$-realisability} \\
\hline
$|\Procs|\leq 3$ & $\gt:det\mapsto \dualdfaof{\gt}:det$ & linear & PSPACE \\
\hline
commutation-closed & $\gt:det\mapsto \dualdfaof{\gt}:det$ & linear & PSPACE \\
\hline
$\acommunicationmodel$-realisable & $\gt:ndet\mapsto \dualdfaof{\detof{\productof{\projectionof{\gt}}}}:det$ & \begin{tabular}{c}doubly\\ exponential \end{tabular}& constant time \\
\hline
\begin{tabular}{c}sender-driven \\ choices \end{tabular}
& $\gt:det\mapsto \renun{\gt}:ndet$ & linear & PSPACE \\
\hline
\begin{tabular}{c}commutation \\ deterministic \end{tabular}& $\gt:det\mapsto \renun{\gt}:ndet$ & linear & PSPACE \\
\hline
\end{tabular}
}
\end{center}
\caption{Summary of complementation procedures for global types}\label{fig:complementation-procedures-summary}
\end{figure*}  

We believe that our results not only improve theoretical clarity but also are amenable to generalisation. This paper represents a first step towards a parametric framework for realisability checking, with potential applicability across diverse communication models. Throughout the paper, we have made a deliberate effort to explicitly state the hypotheses (related to the communication model) under which our theorems hold. 
A particularly critical assumption is that the communication models in question must be causally closed (see Definition \ref{def:causally-closed-communication-model}). Based on this, we conjecture that our results can be extended to models such as bag and causally ordered systems, both of which satisfy causal closure. In contrast, models like mailbox  or those based on bounded FIFO channels lack this property and therefore require specialised analysis. We aim to explore these cases further and, ultimately, to develop a more comprehensive framework in future work.

We introduced the notion of commutation-determinism generalising the notion of sender-driven global types.
It is unclear whether this generalisation may be immediately useful, like capturing "real" global types. However, we believe it is a valuable idea as it may prompt further generalisations;
for instance, a reader influenced by categories may think about "time reversing" transformations, and
how global types that are "commutation-codeterministic" (their mirror language is recognized by a commutation-deterministic automaton)
may be complemented.

\iflong
Apart from these generalisations,  we aim to extend the framework to cover more expressive classes of global types, including those that go beyond synchronisable MSCs. We also plan to revisit results on the realisability of MSCs, possibly linking criteria such as send-validity and receive-validity \cite{DBLP:conf/cav/LiSWZ23,DBLP:conf/ecoop/Stutz23} to our synchronous realisability conditions. This opens the door to formally investigating whether global types that are realisable in the synchronous model satisfy send-validity, or if this is a weaker or stronger condition—this remains an open and promising question.
\fi

% Our work brings that line of research into the realm of programming language design with MPST: we build upon the insights of MSC realisability (e.g., the importance of safety conditions for realisability) and apply them in a type-theoretic framework to obtain results directly applicable to real programming abstractions for concurrency. 
% In summary, by focusing on $\ppmodel$ and synchronous communication models, we shed new light on the implementability of multiparty protocols in practical distributed programming, overcoming the limitations of sender-driven choice and advancing the state of the art in MPST-based program verification.

\subsubsection*{\bf Related works}
We conclude with a brief excursus on related works. Naturally, the most closely related lines of research are those concerning MPST. Comprehensive surveys on the topic are available in \cite{Coppo2015, DBLP:conf/icdcit/YoshidaG20}.
It is noteworthy that in standard MPST literature, realisability—often referred to as implementability—is typically addressed through syntactic restrictions on global types. In most cases, a global type is  implementable if a projection function exists. These syntactic constraints, along with the existence of the projection, imply that the global type is sender-driven, thus placing it within one of the complementable classes discussed in this paper.
In contrast, a recent paper by Scalas and Yoshida~\cite{DBLP:journals/pacmpl/ScalasY19} offers a novel perspective. Their approach emphasizes session fidelity and deliberately abandons syntactic restrictions on session types. Instead, they adopt semantic criteria—validated through model checking—to establish implementability, marking a significant shift from the conventional syntactic paradigm or our approach.
Although, the comparison is difficult as the underlying models are substantially different, it is worth mentioning \cite{10.1145/3678232.3678245} that discusses projectability in the context of synchronous MPST. 

The connection between communicating automata and behavioural types were 
first explored by Villard \cite{villard-phd} for the binary case and Deniélou and Yoshida in \cite{DBLP:conf/esop/DenielouY12} for MPSTs; the connection between HMSCs and MPST was
initiated by Stutz~\cite{DBLP:conf/ecoop/Stutz23,stutz-phd}. In this work, we slightly push further this
connection expliciting the MSC-based semantics of global types that remained somehow implicit in Stutz's work.

In~\cite{DBLP:conf/cav/LiSWZ23,DBLP:conf/ecoop/Stutz23}, the authors revisited
the theory of MPST projections through the lenses of realisability
of MSCs, and proposed two criteria, called \emph{send validity} and \emph{receive validity},
that yield the first sound and complete
characterisation of MPST implementable in the context of
the $\ppmodel$ communication model, provided the global type is sender-driven.
In particular, the second criteria, receive validity,
is strongly justified by the choice of the communication model and is unsound for
other communication models beyond $\ppmodel$. 
An important difference, though, is that these works address the \emph{implementability} problem,
that defines deadlock-freedom as the absence of situations where all processes are blocked. Unlike
deadlock-free realisability, implementability
does not guarantee the absence of unspecified receptions or orphan messages immediately from its definition, 
but only induces it "as a side effect", at least for sender-driven global types. 
 
Realisability was introduced by Alur~\emph{et al} \cite{DBLP:journals/tcs/AlurEY05} with early undecidability results, later simplified by Lohrey \cite{DBLP:journals/tcs/Lohrey03}, showing in particular
that deadlock-free realisability is decidable for bounded communications (a result recently generalised by Bollig~\emph{et al}~\cite{DBLP:conf/apn/BolligFG25}), but undecidable with 5 participants. Our results show that  realisability is decidable with 3 participants, which leaves a gap, even if Alur~\emph{et al} showed that \emph{weak} realisability is undecidable for 4 participants with
$\ppmodel$ communications. 
Guanciale and Tuosto \cite{DBLP:journals/jlap/GuancialeT19} studied realisability on pomsets.

Mazurkiewicz traces~\cite{DBLP:books/ws/95/DR1995} are words over partially commutative alphabets. Synchronous MSCs (but not asynchronous ones) can be seen as Mazurkiewicz traces over the alphabet $\Arrows$ of arrows. Consistent regular trace languages~\cite{AALBERSBERG19881} are the ones recognized by commutation-closed finite automata. Zielonka's asynchronous automata~\cite{DBLP:conf/ershov/Zielonka89} are a model of a distributed implementation of a regular trace language. It differs from realizability of MPSTs and HMSCs in several ways: first, processes in this model synchronise by rendezvous (that may involve more than two processes), second processes may exchange additionnal information during rendezvous (not just the message label specified by the choreography), and third, the choreography is commutation-closed. As a consequence, all regular trace languages are implementable, but the "projection" operation (the "gossip") is highly non-trivial in this setting. Although quite
far from MPSTs at first sight, some old works on Mazurkiewicz traces could possibly shed some light on a few questions that arise from our work, and which, to the best of our knowledge, remain open, like the decidability of the complementability of a global type, or concise complementation procedures for non commutation-closed global types beyond commutation-determinism.

\bibliographystyle{plainurl}
\bibliography{refs}

\appendix
% !TEX root =  ../mainPPDP.tex
%!TEX spellcheck = en_GB

\section{Standard notions on automata}
\label{app:automata-standards}
A non-deterministic finite automaton (NFA) is a tuple $\A=(Q,\Sigma,\delta,q_0,F)$, where $Q$ is a finite set of states, $\Sigma$ is a finite alphabet, $\delta: Q\times(\Sigma\cup \varepsilon)\times Q$ is the transition relation, $q_0\in Q$ is the initial state, and $F\subseteq Q$ is the set of accepting states. We write $\delta^*(s,w)$ to denote the set of states $s'$ that are reachable from $s$ following a path labeled with $w$. The language accepted by $\A$, denoted $\languageofnfa{\A}$, is the set of words $w\in\Sigma^*$ such that $\delta^*(q_0,w)\cap F\neq\emptyset$..

A deterministic finite automaton (DFA) is a special case of an NFA where the transition relation $\delta$ is a partial function
$\delta: Q\times\Sigma\to Q$; it is complete if the function is
total. Every DFA $\A=(Q,\Sigma,\delta,q_0,F)$ 
accepts the same language as the complete DFA
$\A'=(Q\cup\{\bot\},\Sigma,\delta',q_0,F)$ where 
$\delta'(q,a)=\bot$ whenever $\delta(q,a)$ is undefined, and $\delta'(\bot,a)=\bot$.

To every NFA $\A=(Q,\Sigma,\delta,q_0,F)$, one can associate a DFA $\detof{\A}\eqdef(Q',\Sigma,\delta',q_0',F')$ where $Q'=2^Q$, $q_0'=\{q_0\}$, $F'$ is the set of subsets of $Q$ that contain at least one accepting state, and $\delta'$ is defined such that for all $S\in Q'$ and $a\in\Sigma$, $\delta'(S,a)=\bigcup\{\delta^*(s,a)\mid s\in S\}$.

Given two NFAs $\A_1=(Q_1,\Sigma,\delta_1,q_{0,1},F_1)$ and $\A_2=(Q_2,\Sigma,\delta_2,q_{0,2},F_2)$, we define the product NFA 
$$\A_1\otimes\A_2~\eqdef~(Q_1\times Q_2,\Sigma,\delta,(q_{0,1}, q_{0,2}),F_1\times F_2)$$ 

where the transition relation $\delta$ is defined as follows: 

\begin{itemize}
    \item for all $(s_1,s_2)\in Q_1\times Q_2$, $a\in\Sigma$, and $(s'_1,s'_2)\in Q_1\times Q_2$, we have $((s_1,s_2),a,(s'_1,s'_2))\in\delta$ if and only if $(s_1,a,s'_1)\in\delta_1$ and $(s_2,a,s'_2)\in\delta_2$, and
    \item for all $(s_1,s_2)\in Q_1\times Q_2$ and $(s'_1,s'_2)\in Q_1\times Q_2$, we have $((s_1,s_2),\varepsilon,(s'_1,s'_2))\in\delta$ if and only if $(s_1,\varepsilon,s'_1)\in\delta_1$ or $(s_2,\varepsilon,s'_2)\in\delta_2$.
\end{itemize}
It holds that $\languageofnfa{\A_1\otimes\A_2}=\languageofnfa{\A_1}\cap \languageofnfa{\A_2}$.

Finally, if $\A=(Q,\Sigma,\delta,q_0,F)$ is a DFA, we can define its dual DFA $\dualdfaof{\A}\eqdef(Q,\Sigma,\delta,q_0,Q\setminus F)$. In this case, if $\A$ is complete, we have $\languageofnfa{\dualdfaof{\A}}=\Sigma^*\setminus\languageofnfa{\A}$.

\section{Proof of Proposition~\ref{prop:deadlock-free-as-a-property-on-mscs-for-p2p-and-synch}}
\label{app:deadlock-free-as-a-property-on-mscs-for-p2p-and-synch}
\propdeadlockfreeasapropertyonmscsforppandsynch*
% !TEX root =  ../mainPPDP.tex
%!TEX spellcheck = en_GB

\begin{proof}
    We show each implication separately.
    \begin{description}
        \item[$\Rightarrow$:]
        if $\cfsms$ is deadlock-free: $\executionsof{\acceptcompletion{\cfsms}}{\acommunicationmodel}\subseteq\prefixclosureof{\executionsof{\cfsms}{\acommunicationmodel}}$, 
        then 
        $\msclanguageof{\acceptcompletion{\cfsms}}{\acommunicationmodel}\subseteq
        \prefixclosureof{\msclanguageof{\cfsms}{\acommunicationmodel}}$.
        For this implication, $\acommunicationmodel$ can be any communication model.
        Let $\msc\in\msclanguageof{\acceptcompletion{\cfsms}}{\acommunicationmodel}$, 
       we show that $\msc\in\prefixclosureof{\msclanguageof{\cfsms}{\acommunicationmodel}}$.
        By definition, there is an execution $e\in\executionsof{\acceptcompletion{\cfsms}}{\acommunicationmodel}$ such that $\msc=\mscof{e}$.
        By hypothesis, there is a completion $e' \in\executionsof{\cfsms}{\acommunicationmodel}$ of $e$.
        By definition, $\mscof{e'}\in\msclanguageof{\acceptcompletion{\cfsms}}{\acommunicationmodel}$,
        so $\msc\in\prefixclosureof{\msclanguageof{\cfsms}{\acommunicationmodel}}$.

        \item[$\Leftarrow$:] for  $\acommunicationmodel=\ppmodel$:
        if  
        $\msclanguageof{\acceptcompletion{\cfsms}}{\acommunicationmodel}\subseteq
        \prefixclosureof{\msclanguageof{\cfsms}{\acommunicationmodel}}$,
        then $\cfsms$ is deadlock-free.
        Let $e\in\executionsof{\acceptcompletion{\cfsms}}{\acommunicationmodel}$ be an execution,
        wez show that $e$ has a completion in $\executionsof{\cfsms}{\acommunicationmodel}$.
        By definition, $\mscof{e}\in\msclanguageof{\acceptcompletion{\cfsms}}{\acommunicationmodel}$.
        By hypothesis, there is a MSC $\msc'\in\msclanguageof{\cfsms}{\acommunicationmodel}$
        such that $\mscof{e}\prefixordermsc\msc'$. By definition of $\prefixordermsc$ on MSCs,
        there are two executions $e_1,e'$ such that 
        $e_1\prefixorder e'$, $\mscof{e'}=\msc'$, 
        and $\mscof{e_1}=\mscof{e}$.
        Let $<$ be the binary relation on $\eventsof{M'}$ 
        defined as $< \eqdef \torderstrictof{e}\cup <_1\cup <_2$, where:
        \begin{itemize}
            \item $\event <_1<\event'$ if $\event\not\in\eventsof{M}$,
            $\event'\not\in\eventsof{M}$ and $\event <_{e'} \event'$;
            \item $\event <_2\event'$ if $\event\in\eventsof{M}$ and $\event'\not\in\eventsof{M}$.
        \end{itemize}
        Note that $<_1$ and $<_2$ are transitive. We claim that $<$
        is also transitive. Indeed, 
        $$
        \begin{array}{rcl}
           \torderof{e}\cdot <_1 & = & <_1\\
            <_1\cdot <_2 & = & <_1 \\
            \torderof{e}\cdot<_2 & = & \emptyset \\
            (<_1\cup <_2)\cdot \torderof{e} & = & \emptyset \\
            (\torderof{e}\cup <_2)\cdot <_1 & = & \emptyset \\
        \end{array}
        $$
        Moreover, $<$ is irreflexive, because $\torderof{e}$ is irreflexive and $<_1$ and $<_2$ are irreflexive.
        Let $\leq$ denote the reflexive closure of $<$.
        Then this is a partial order on $\eventsof{M'}$.
        By construction, $\leq$ is actually a total order on $\eventsof{M'}$.       
        Finally, we claim that $\leq$ is a linearisation of $M'$, i.e., $\porderof{M'}\subseteq \leq$.
        Indeed, assume that $\event\porderstrictof{M'}\event'$, we show that $\event< \event'$:
        \begin{itemize}
            \item if $\event\in\eventsof{M}$ and $\event'\in\eventsof{M}$,
            then $\event\porderstrictof{M} \event'$ (as $M$ is a prefix of $M'$), hence $\event\torderof{e} \event'$ 
            (because $e$ is a linearisation of $M$),
            and finally $\event < \event'$ by definition of $<$.
            
            \item if $\event\not\in\eventsof{M}$ and $\event'\not\in\eventsof{M}$,
            then $\event\torderof{e'} \event'$ (because $e'$ is a linearisation of $M'$), 
            therefore $\event <_1 \event'$ by definition of $<_1$, and finally $\event < \event'$ by definition of $<$.

            \item if $\event\in\eventsof{M}$ and $\event'\not\in\eventsof{M'}$,
            then $\event<_2 \event'$ by definition of $<_2$, therefore $\event < \event'$ by definition of $<$.
        \end{itemize}
        Therefore, $\leq$ is a linearisation of $M'$. Let $e''$ be the execution associated to $\leq$.
        Then
        \begin{itemize}
            \item $e\prefixorder e''$ by definition of $\leq$;
            \item $\mscof{e''}=M'$ as $\leq$ is a linearisation of $M'$;
            \item $e''\in\executionsof{\cfsms}{\acommunicationmodel}$ because $e''$ is a linearisation of $M'$, $M'\in\msclanguageof{\cfsms}{\acommunicationmodel}$,
            and $\ppmodel$ is causally closed (Lemma~\ref{lem:pp-is-causally-closed}).
        \end{itemize}
        Altogether, we have shown that $e$ has a completion in $\executionsof{\cfsms}{\acommunicationmodel}$, which ends the proof of 
        the reverse implication.

        \item[$\Leftarrow$:]   for $\acommunicationmodel=\synchmodel$:
            the proof is similar to previous case. We define the linearisation $\leq$ of $M'$ 
            in the exactly  same way, but we cannot use the fact that $\acommunicationmodel$ is causally closed.
            Instead, we use the fact that both $<_e$ and  $<_1$ are synchronous linearisations,
            and therefore $<$ is a synchronous linearisation as the concatenation of two synchronous linearisations.

    \end{description}
\end{proof}

\section{Proof of Lemma~\ref{lem:product-of-gt}}
\label{app:product-of-gt}
\lemproductofgt*
\begin{proof}
    We reason by double inclusion.
    \begin{itemize}
        \item[$\Rightarrow$] Assume that $M\in\existentialmsclanguageof{\gt_1\otimes\gt_2}$.
        By definition, there is a word $w$ such that $M=\mscof{w}$, and $w\in\languageofnfa{\gt_1\otimes\gt_2}$.
        By the definition of the product of global types, $\languageofnfa{\gt_1\otimes\gt_2}=\languageofnfa{\gt_1}\cap\languageofnfa{\gt_2}$.
        In particular, $w\in\languageofnfa{\gt_1}$, thus $\mscof{w}\in\existentialmsclanguageof{\gt_1}$.
        Similarly, $w\in\languageofnfa{\gt_2}$, thus $\mscof{w}\in\existentialmsclanguageof{\gt_2}$.
        Since $\gt_2$ is commutation-closed, $\mscof{w}\in\msclanguageofcc{\gt_2}$.
        Therefore, $\mscof{w}\in\existentialmsclanguageof{\gt_1}\cap\msclanguageofcc{\gt_2}$.
        \item[$\Leftarrow$] Conversely, assume that $M\in\existentialmsclanguageof{\gt_1}\cap\msclanguageofcc{\gt_2}$.
        In particular, $M\in\existentialmsclanguageof{\gt_1}$, which means that there is a word $w_1$ such that 
        $M=\mscof{w_1}$ and $w_1\in\languageofnfa{\gt_1}$.
        Since $M\in\msclanguageofcc{\gt_2}$, there is a word $w_2$ such that $M=\mscof{w_2}$ and $w_2\in\languageofnfa{\gt_2}$.
        Since $\gt_2$ is commutation-closed and $\mscof{w_1}=\mscof{w_2}$, $w_1\in \languageofnfa{\gt_2}$.
        Therefore, $w_1\in\languageofnfa{\gt_1}\cap\languageofnfa{\gt_2}=\languageofnfa{\gt_1\otimes\gt_2}$, 
        which means that $M\in\existentialmsclanguageof{\gt_1\otimes\gt_2}$.
    \end{itemize}
\end{proof}

\section{Proof of Theorem~\ref{thm:main-theorem-realisability}}
\label{app:main-theorem-realisability}
\maintheoremrealisability*

% \section{Discussion about deadlock-free realizability and safe realizability}
% \label{app:discussion-deadlock-free-realizability-and-safe-realizability}
% \input{content/discussion-deadlock-free-realizability-and-safe-realizability.tex}

\end{document}